\documentclass[12pt, draft, onecolumn]{IEEEtran}

\usepackage{cite}
\usepackage[textwidth=6.53in, textheight=24.8cm]{geometry}
\usepackage{amsthm, amsmath,amssymb,amsfonts,bbm}
\usepackage{algorithmic}
\usepackage[dvips]{graphicx}
\usepackage{subfigure}
\graphicspath{{graphics/}}
\usepackage{makecell}
\usepackage{multirow}
\usepackage{filecontents, pgfplots}
\usepackage[algoruled]{algorithm2e}
\usepackage{microtype}
\usepackage{tikz}
\usepackage{textcomp}
\usepackage{xcolor}
\usepackage{verbatim}
\usepackage[binary-units=true]{siunitx}

\DeclareSIUnit{\belmilliwatt}{Bm}
\DeclareSIUnit{\dBm}{\deci\belmilliwatt}
\def\BibTeX{{\rm B\kern-.05em{\sc i\kern-.025em b}\kern-.08em
		T\kern-.1667em\lower.7ex\hbox{E}\kern-.125emX}}

\makeatletter
\newif\iftag@here

\newcommand*{\taghere}[1][0pt]
{\ifmeasuring@\else
	\global\tag@heretrue
	\tikz[remember picture,overlay]{\coordinate (taghere) at (0pt,#1);}%
	\fi}

\def\place@tag{%
	\iftagsleft@
	\kern-\tagshift@
	\iftag@here
	\global\tag@herefalse
	\tikz[remember picture,overlay]%
	{\path (taghere) -| node[anchor=base]{\rlap{\boxz@}} (0pt,0pt);}%
	\else
	\if1\shift@tag\row@\relax
	\rlap{\vbox{%
			\normalbaselines
			\boxz@
			\vbox to\lineht@{}%
			\raise@tag
	}}%
	\else
	\rlap{\boxz@}%
	\fi
	\kern\displaywidth@
	\fi
	\else
	\kern-\tagshift@
	\iftag@here
	\global\tag@herefalse
	\tikz[remember picture,overlay]%
	{\path  (taghere) -|  node[anchor=base]{\llap{\boxz@}} (0pt,0pt);}%
	\else
	\if1\shift@tag\row@\relax
	\llap{\vtop{%
			\raise@tag
			\normalbaselines
			\setbox\@ne\null
			\dp\@ne\lineht@
			\box\@ne
			\boxz@
	}}%
	\else \llap{\boxz@}%
	\fi
	\fi
	\fi
}
\makeatother

\DeclareMathOperator*{\maximize}{maximize}
\DeclareMathOperator*{\argmin}{arg\,min}
\DeclareMathOperator*{\minimize}{minimize}
\DeclareMathOperator*{\subjectto}{subject\,to}

\usepackage[acronym,nomain]{glossaries}
\makeglossaries
\newacronym{swipt}{SWIPT}{simultaneous wireless information and power transfer}
\newacronym{wpt}{WPT}{wireless power transfer}
\newacronym{wpcn}{WPCN}{wireless powered communication network}
\glsunset{wpcn}
\newacronym{wit}{WIT}{wireless information transfer}
\newacronym{awgn}{AWGN}{additive white Gaussian noise}
\newacronym{tx}{TX}{transmitter}
\newacronym{bs}{BS}{base station}
\newacronym{ir}{IR}{information receiver}
\newacronym{eh}{EH}{energy harvesting}
\newacronym{irs}{IRS}{intelligent reflected surface}
\newacronym{ap}{AP}{average power}
\newacronym{pp}{PP}{peak power}

\newacronym{siso}{SISO}{single-input single-output}
\newacronym{mimo}{MIMO}{multiple-input multiple-output}
\newacronym{miso}{MISO}{multiple-input single-output}
\newacronym{simo}{SIMO}{single-input multiple-output}

\newacronym{rf}{RF}{radio frequency}
\newacronym{dc}{DC}{direct current}
\newacronym{ac}{AC}{alternative current}
\newacronym{papr}{PAPR}{peak-to-average power ratio}
\newacronym{lp}{LPF}{low-pass filter}
\newacronym{mc}{MC}{matching circuit}
\newacronym{mrt}{MRT}{maximum ratio transmission}
\newacronym{ecb}{ECB}{equivalent complex baseband}
\newacronym{tdd}{TDD}{time-division-duplex}
\newacronym{zf}{ZF}{zero forcing}
\newacronym{snr}{SNR}{signal-to-noise ratio}
\newacronym{sinr}{SINR}{signal-to-interference-plus-noise ratio}
\glsunset{sinr}
\newacronym{tdma}{TDMA}{time division multiple access}

\newacronym{rv}{RV}{random variable}
\newacronym{iid}{i.i.d.}{independent and identically distributed}
\newacronym{pdf}{pdf}{probability density function}

\newacronym{dnn}{DNN}{dense neural networks}
\glsunset{dnn}
\newacronym{mdp}{MDP}{Markov decision process}

\newacronym{sca}{SCA}{successive convex approximation}
\newacronym{sdr}{SDR}{semi-definite relaxation}

\newacronym{spr}{LP}{low power}
\newacronym{mpr}{MP}{medium power}
\newacronym{lpr}{HP}{high power}

\DeclareMathOperator{\rank}{rank}
\DeclareMathOperator{\diag}{diag}
\newcommand{\norm}[1]{\left\lVert#1\right\rVert_2}
\newcommand{\Tr}[1]{\text{Tr}\{#1\} }
\begin{document}

	\newtheorem{proposition}{Proposition}	
	\newtheorem{lemma}{Lemma}	
	\newtheorem{corollary}{Corollary}
	\newtheorem{assumption}{Assumption}	
	\newtheorem{remark}{Remark}	
	
	\title{Optimal Energy Signal Design for Multi-user MISO WPCNs With Non-linear Energy Harvesting Circuits}

	\author{\IEEEauthorblockN{Nikita Shanin\IEEEauthorrefmark{1}, Amelie Hagelauer\IEEEauthorrefmark{2}\IEEEauthorrefmark{3}, Laura Cottatellucci\IEEEauthorrefmark{1}, and Robert Schober\IEEEauthorrefmark{1}}
		\IEEEauthorblockA{\IEEEauthorrefmark{1}\textit{Friedrich-Alexander-Universit\"{a}t Erlangen-N\"{u}rnberg (FAU), Germany}\\ 
		\IEEEauthorrefmark{2}\textit{Fraunhofer EMFT Einrichtung f\"{u}r Mikrosysteme und Festk\"{o}rper-Technologien}\\
		\IEEEauthorrefmark{3}\textit{Technische Universit\"{a}t M\"{u}nchen, Germany}}	}

	\maketitle
	\vspace*{-55pt}
	\begin{abstract}
		\vspace*{-5pt}
		\let\thefootnote\relax\footnotetext{This paper was presented in part at IEEE International Conference on Acoustics, Speech and Signal Processing (ICASSP), Singapore, 2022 \cite{Shanin2021d}. }
In this work, we study a multi-user wireless powered communication network (WPCN), where a multi-antenna base station (BS) sends an energy signal to multiple single-antenna users equipped with non-linear energy harvesting (EH) circuits.
The users, in turn, harvest energy from the received signal and utilize it for information transmission in the uplink.
Furthermore, to jointly optimize the energy signal waveform and downlink beamforming, we assume that the BS broadcasts a pulse-modulated signal employing multiple energy signal vectors. 
We formulate an optimization problem for the joint design of the downlink transmit energy signal vectors, their number, the durations of the transmit pulses, and the time allocation policy for minimization of the average transmit power at the BS.
We show that for single-user WPCNs, a single energy signal vector, which is collinear with the maximum ratio transmission (MRT) vector and drives the EH circuit at the user device into saturation, is optimal.
Next, for the general multi-user case, we show that the optimal signal design requires a maximum number of energy signal vectors that exceeds the number of users by one and propose an algorithm to obtain the optimal energy signal vectors.
Since the complexity of the optimal design is high, we also propose two suboptimal schemes for WPCN design.
First, for asymptotic massive WPCNs, where the ratio of the number of users to the number of BS antennas, i.e., the system load, tends to zero, we show that the optimal downlink transmit signal can be obtained in closed-form and comprises a sequence of weighted sums of MRT vectors. 
Next, based on this result, for general WPCNs with finite system loads, we propose a suboptimal closed-form MRT-based design and a suboptimal semidefinite relaxation (SDR)-based scheme.
Our simulation results reveal that the proposed optimal scheme and suboptimal SDR-based design achieve nearly identical performance and outperform two baseline schemes, which are based on linear and sigmoidal EH models.
Furthermore, we show that, if the system load of the WPCN is low, the performance gap between the proposed suboptimal solutions is small and becomes negligible as the number of BS antennas tends to infinity. 

	\end{abstract}

\section{Introduction}
The growth of the number of low-power Internet-of-Things (IoT) devices has fuelled significant interest in wireless powered communication networks (WPCNs) that enable energy-sustainable communication \cite{Bi2015, Clerckx2019, Ju2014, Liu2014, Yang2015, Lee_2016, Morsi2018a, Gu2022, Wei2022, Tietze2012, Shanin2020, Boshkovska2015, Boshkovska2017a, Boshkovska2018, Hua2022, Zeng2021, Li2022, Clerckx2015, Clerckx2018, Clerckx2017, Zawawi2019, Shanin2021a, Morsi2019, Shanin2021d}.
A typical multi-user \gls*{miso} WPCN comprises a multi-antenna \gls*{bs} that broadcasts a \gls*{rf} signal to multiple single-antenna users in the downlink \cite{Bi2015}. 
Each user device is equipped with an electrical circuit for {harvesting} the received \gls*{rf} power and storing it in a load device (e.g., a battery).
Then, the users employ the harvested power for information transmission in the uplink \cite{Clerckx2019}.

For the design of WPCNs, a linear relationship between the received and harvested powers at the user devices is assumed in \cite{Ju2014, Liu2014, Yang2015, Lee_2016, Morsi2018a}.
In \cite{Ju2014}, for multi-user \gls*{siso} WPCNs, the authors formulate a sum-throughput maximization problem and unveil a {\itshape double-near-far} phenomenon resulting in an uneven rate allocation among the users.
To ensure reliable uplink communication, the authors in \cite{Ju2014} also considered a minimum-throughput maximization problem and showed that, due to the distance-dependent signal attenuation, the \gls*{bs} has to allocate significantly more resources in the downlink to far users compared to near users.
Next, in \cite{Liu2014}, the authors study multi-user WPCNs, where the \gls*{bs} employs multiple antennas and utilizes \gls*{zf} equalization to suppress the inter-user interference in the uplink.
For this setup, the authors propose an algorithm for downlink energy beamforming that maximizes the minimum information rate in the uplink.
The authors in \cite{Yang2015} study a massive \gls*{mimo} WPCN and show that the optimal energy beamforming in the downlink is a linear combination of the normalized channel vectors between the \gls*{bs} and the user devices.
In \cite{Lee_2016}, the authors consider multi-user WPCNs, where \gls*{tdma} is adopted for information transmission in the uplink and users are equipped with infinite or finite capacity energy storage.
It is shown that the optimal transmit policy at the BS is to radiate all available power as fast as possible, and then, remain silent until the end of the scheduling time frame. 
The energy storage at the user devices of a WPCN is further considered in \cite{Morsi2018a}.
In this paper, the authors characterize the internal state of the energy buffer via a Markov chain and, for Rayleigh fading channels, determine the limiting distribution of the stored power in closed-form.

Although the results in \cite{Ju2014, Liu2014, Yang2015, Lee_2016, Morsi2018a} provide notable insights for WPCN design, they are based on restrictive impractical assumptions.
In fact, practical \gls*{eh} circuits employ non-linear diodes for signal rectification and thus, the dependence of the harvested power on the input power is highly non-linear in both the low and high input power regimes \cite{Gu2022, Wei2022}.
If the received \gls*{rf} power at a user device is low, the non-linearity is caused by both the non-linear current-voltage (I-V) characteristic of the rectifying diode \cite{Tietze2012} and the power-dependent impedance mismatch between the receive antenna and the non-linear rectifier circuit \cite{Shanin2020}.
For high input powers, the non-linear diode is driven into breakdown, which causes saturation of the harvested power at the user device \cite{Clerckx2019}.
To take these non-linear effects of the EH circuits into account, the authors in \cite{Boshkovska2015} propose a sigmoid function-based \gls*{eh} model, whose parameters are obtained via curve fitting to match simulation data.
This model is widely employed for the design of WPCNs, see, e.g., \cite{Boshkovska2017a, Boshkovska2018, Hua2022, Zeng2021, Li2022}.
In particular, the authors in \cite{Boshkovska2017a} study multi-user \gls*{mimo} \gls*{wpcn}s, where, similarly to \cite{Lee_2016}, a \gls*{tdma} scheme is employed for information transfer in the uplink.
In \cite{Boshkovska2018}, the authors consider a wireless communication system, where a wirelessly charged multi-antenna \gls*{bs} transmits information to multiple single-antenna receivers in the presence of an eavesdropper.
Furthermore, the authors in \cite{Hua2022, Zeng2021, Li2022} study WPCNs assisted by intelligent reflecting surfaces, which are employed to increase the harvested powers at the user devices, and thus, improve the overall uplink throughput.
For the considered system setups, the authors in \cite{Boshkovska2017a, Boshkovska2018, Hua2022, Zeng2021, Li2022} formulate non-convex resource allocation optimization problems for system design and propose iterative algorithms to solve them.

The analysis in \cite{Boshkovska2017a, Boshkovska2018, Hua2022, Zeng2021, Li2022} represents a significant progress over the results in \cite{Ju2014, Liu2014, Yang2015, Lee_2016, Morsi2018a}.
However, the model introduced in \cite{Boshkovska2015} and adopted in \cite{Boshkovska2017a, Boshkovska2018, Hua2022, Zeng2021, Li2022} characterizes the {\itshape average} harvested power at the user devices as a function of the {\itshape average} received RF power, and therefore, cannot fully capture the non-linearity of EH circuits.
Interestingly, the authors in \cite{Clerckx2015} propose an EH model that maps the {\itshape instantaneous} received RF power to the {\itshape instantaneous} harvested power. 
Based on this EH model, the author in \cite{Clerckx2018} considers single-user wireless power transfer systems and shows that the power harvested at the user is maximized if \gls*{mrt} is adopted at the multi-antenna \gls*{bs} in the downlink.
Furthermore, the EH model in \cite{Clerckx2015} is utilized in \cite{Clerckx2017} and \cite{Zawawi2019} for the analysis of backscatter communication systems which are also based on wireless power transfer.
In particular, the authors in \cite{Clerckx2017} highlight a trade-off between the harvested power at the user device and the \gls*{snr} of the communication link between user and information receiver.
In \cite{Zawawi2019}, the authors extend the communication system in \cite{Clerckx2017} to the multi-user case and formulate an optimization problem to analyze the trade-off between the weighted sum of harvested powers at the user devices and the signal-to-interference-plus-noise ratios (SINRs) of the communication links between the users and the receiver.

In contrast to the linear and sigmoidal models in \cite{Ju2014, Liu2014, Yang2015, Lee_2016, Morsi2018a} and \cite{Boshkovska2015, Boshkovska2017a, Boshkovska2018, Hua2022, Zeng2021, Li2022}, respectively, the EH model in \cite{Clerckx2015} characterizes the instantaneous behaviour of the EH circuits and thus, allows the optimization of the transmit signal waveform for wireless power transfer.
However, since this model is based on the Taylor series approximation of the current flow through the rectifying diode and does not take into account the breakdown effect of the diode, it still does not fully capture all non-linear effects of practical EH circuits for both low and high input powers.
To overcome this limitation, the authors in \cite{Morsi2019} propose an EH model based on a precise analysis of a standard EH circuit with half-wave rectifier and show that On-Off transmission is optimal for the maximization of the average harvested power in single-antenna wireless power transfer systems.
Next, since the model in \cite{Morsi2019} is valid for half-wave rectifying EH circuits only, the authors in \cite{Shanin2020} propose a learning-based approach to model a wide class of EH circuits and show that the behaviour of an EH circuit depends on the number of rectifying diodes, the characteristics of the matching circuit, and other circuit parameters.
The authors in \cite{Shanin2021a} study a multi-user \gls*{mimo} wireless power transfer system with general non-linear EH circuits.
They design the optimal energy signal to be transmitted by the BS to maximize a weighted sum of the average harvested powers at the EH nodes and show that it comprises multiple beamforming vectors.
Finally, in \cite{Shanin2021d}, which is the conference version of this paper, we adopt the EH model from \cite{Morsi2019} and develop an iterative algorithm for the design of two-user \gls*{miso} WPCNs.
However, to the best of the authors' knowledge, the optimal design of the downlink energy signal waveform and the resource allocation for multi-user MISO WPCNs taking into account all non-linear effects of \gls*{eh} circuits is still an open problem, which is tackled in this paper.

In this paper, we determine the optimal downlink energy signal design and resource allocation policy for multi-user MISO WPCNs, where the single-antenna user devices are equipped with non-linear \gls*{eh} circuits, to minimize the average transmit power at the BS.
In contrast to the WPCN designs in \cite{Ju2014, Liu2014, Yang2015, Lee_2016, Morsi2018a, Boshkovska2015, Boshkovska2017a, Boshkovska2018, Hua2022, Zeng2021, Li2022, Clerckx2015}, in this work, to account for the non-linearities of the EH circuits in both the low and high input power regimes, we adopt a general non-linear EH model that characterizes the instantaneous harvested power at the user devices.
Furthermore, in contrast to the existing works on multi-user WPCNs with multi-antenna BSs, to jointly optimize beamforming and the waveform of the downlink energy signal, we assume that the BS broadcasts a pulse-modulated energy signal employing multiple transmit energy signal vectors. 
The user devices, in turn, harvest the power from the received \gls*{rf} energy signal and utilize it for information transmission in the uplink, where we adopt \gls*{zf} equalization at the \gls*{bs} to suppress inter-user interference.
The main contributions of this paper can be summarized as follows:
\begin{itemize}
	\item We formulate an optimization problem for the joint design of the normalized durations of the downlink and uplink transmission subframes and the number, durations, values, and powers of the downlink transmit energy signal vectors for minimization of the average transmit power at the \gls*{bs} under per-user rate constraints in the uplink.
	\item First, as a special case, we consider a single-user WPCN and show that the optimal energy signal in the downlink employs a single vector, which is collinear with the \gls*{mrt} vector and whose norm is chosen such that it drives the EH circuit at the user device into saturation. 
	Then, for the general multi-user case, we show that the maximum number of transmit energy signal vectors for the optimal signal design in the downlink exceeds the number of users by one.
	To obtain these energy signal vectors, we propose an optimal algorithm whose computational complexity is exponential in the number of user devices and polynomial in the number of antennas equipped at the \gls*{bs}.
	\item Since the optimal solution for the general multi-user case entails a high computational complexity if the number of users is large, we also propose two low-complexity suboptimal solutions.
	First, we show that for asymptotic massive MISO WPCNs, where the ratio of the users and number of BS antennas, i.e., the system load, tends to zero, the optimal downlink energy signal comprises a sequence of weighted sums of the \gls*{mrt} beamforming vectors and can be obtained in closed-form.
	Next, based on this result, we design an MRT-based scheme, which is optimal for massive MISO WPCNs with vanishing system loads and provides a suboptimal solution of the formulated problem for general MISO WPCNs with finite system loads.
	Furthermore, to improve the performance of the MRT-based design for general WPCNs, we derive a low-complexity suboptimal scheme, which is based on \gls*{sdr}.
	\item Our simulation results reveal that the proposed SDR-based suboptimal design has a significantly lower computational complexity than the optimal scheme for high system loads and both WPCN designs achieve nearly identical performance and significantly outperform two baseline schemes, which are based on the linear and sigmoidal EH models, respectively.
	Also, we show that for low system loads, the proposed suboptimal MRT- and SDR-based schemes have a similar performance and become identical when the number of BS antennas tends to infinity.
	Finally, we show that when the number of deployed BS antennas and user devices grow, the average transmit power in the downlink decreases and increases, respectively.
\end{itemize}

The remainder of this paper is organized as follows. 
In Section II, we discuss the proposed system model.
In Section III, we formulate the optimization problem for minimizing the average transmit power in the downlink under per-user rate constraints in the uplink, and solve it for the single-user case in closed-form.
Moreover, for the general multi-user case, we characterize the optimal solution and propose an algorithm to compute it.
In Section IV, we consider the asymptotic massive MISO regime and determine a closed-form solution for the resulting optimization problem.
Furthermore, for the general non-asymptotic case, we propose two low-complexity suboptimal schemes.
In Section V, we evaluate the performance of the proposed framework via numerical simulations.
Finally, in Section VI, we draw some conclusions.

\emph{Notation:} Bold upper case letters $\boldsymbol{X}$ represent matrices and ${X}_{i,j}$ denotes the element of $\boldsymbol{X}$ in row $i$ and column $j$. 
Bold lower case letters $\boldsymbol{x}$ stand for vectors and ${x}_{i}$ is the $i^\text{th}$ element of $\boldsymbol{x}$.
$\boldsymbol{X}^H$, $\Tr{\boldsymbol{X}}$, and $\rank \{\boldsymbol{X}\}$ denote the Hermitian, trace, and rank of matrix $\boldsymbol{X}$, respectively.
$\mathbb{E}\{x\}$ denotes the statistical expectation of $x$.
The real part of a complex number is denoted by $\Re\{ \cdot \}$.
The transpose and L2-norm of vector $\boldsymbol{x}$ are represented by $\boldsymbol{x}^\top$ and $\norm{\boldsymbol{x}}$, respectively.
The sets of real, real non-negative, and complex numbers are represented by $\mathbb{R}$, $\mathbb{R}_{+}$, and $\mathbb{C}$, respectively, whereas $\mathcal{S}^{N}_{+}$ stands for the set of complex positive semidefinite matrices of size $N$.
${\boldsymbol{X}}^{\frac{1}{2}}$ and $\norm{\boldsymbol{X}}$ stand for the square root and spectral norm of matrix $\boldsymbol{X}$, respectively, whereas $\diag(\boldsymbol{x})$ with $\boldsymbol{x} \in \mathbb{C}^N$ represents a diagonal matrix $\boldsymbol{X} \in \mathbb{C}^{N \times N}$, whose elements are all zero-valued except for ${X}_{n,n} = x_n, \forall n \in \{1,2,\cdots, N\}$.
The imaginary unit is denoted by $j$.
$f^{-1}(\cdot)$ and $f'(x_0)$ denote the inverse function of $f(\cdot)$ and the first-order derivative of $f(x)$ evaluated at point $x = x_0$, respectively.
Generalized non-strict inequalities are denoted by $\succeq$ and $\preceq$ and are associated with $\mathbb{R}^N_{+}$, i.e., $\forall N \geq 0$ and $\forall \boldsymbol{x}, \boldsymbol{y} \in \mathbb{R}^N$, $\boldsymbol{x} \succeq \boldsymbol{y}$ and $\boldsymbol{y} \preceq \boldsymbol{x} \Longleftrightarrow y_n \leq x_n, \forall n \in \{1,2,\cdots, N\}$. 
Furthermore, $\boldsymbol{0}$ and $\boldsymbol{I}_N$ stand for the all-zero vector of appropriate dimension and the identity matrix of size $N$, respectively.

\section{System Model}
\begin{figure}[t]
	\centering
	\includegraphics[draft=false, width=1\textwidth]{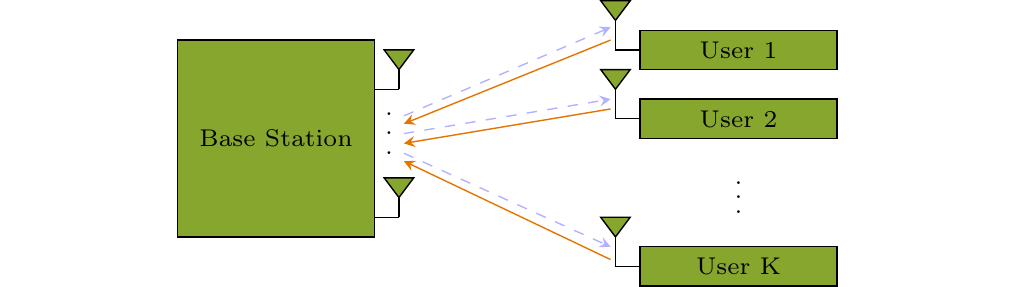}
	\caption{A WPCN with a multi-antenna BS and $K$ single-antenna users. In the downlink, the BS sends an energy signal to the users (blue dashed arrows). In turn, the users harvest the received energy and utilize it for information transmission in the uplink (orange solid arrows).}
	\label{Fig:SystemModel}
	\vspace*{-15pt}
\end{figure}
\begin{figure}[t]
	\centering
	\includegraphics[draft=false, width=1\textwidth]{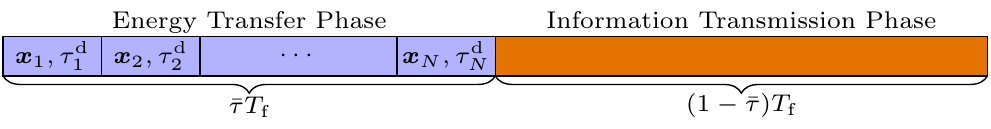}
	\caption{Structure of a time frame of length $T_\text{f}$.}
	\label{Fig:TimeSharing}
	\vspace*{-5pt}
\end{figure}
We consider a multi-user MISO \gls*{wpcn}, where $K \geq 1$ single-antenna users are equipped with \gls*{eh} circuits \cite{Morsi2019} and $N_\text{t}\geq K$ antennas are deployed at the \gls*{bs}, cf. Fig.~\ref{Fig:SystemModel}.
To enable EH at the user devices and information transmission in the uplink, we adopt \gls*{tdd} transmission and assume that each time frame of length $T_\text{f}$ is divided into two subframes. 
In particular, in the first subframe of length $\bar{\tau} T_\text{f}$ with $\bar{\tau} \in [0,1]$, the BS transmits an \gls*{rf} energy signal to transfer energy to the user devices, which harvest the received power.
In the subsequent subframe of length $(1-\bar{\tau}) T_\text{f}$, this harvested power is utilized for information transmission in the uplink.
We assume that the channels between the BS and the user devices are constant for the duration of a time frame.
The channel between the BS and user $k$ is denoted by row vector $\boldsymbol{h}_{k} \in \mathbb{C}^{1 \times N_\text{t}}$.
We assume that channel reciprocity holds and $\boldsymbol{h}_{k}$ is perfectly known\footnotemark\hspace*{0pt} at the BS and user $k$.

\footnotetext{In this work, we investigate the maximum achievable performance of WPCN systems. Therefore, as in, e.g., \cite{Clerckx2015, Clerckx2017, Zawawi2019, Morsi2019, Shanin2021a, Shanin2020}, we assume perfect channel knowledge at the \gls*{bs} and the user devices.}
\subsection{Energy Transfer Phase}
In the energy transfer phase, in contrast to the WPCN designs in \cite{Ju2014, Liu2014, Yang2015, Lee_2016, Morsi2018a, Boshkovska2015, Boshkovska2017a, Boshkovska2018, Hua2022, Zeng2021, Li2022, Clerckx2015, Clerckx2017, Zawawi2019}, to jointly optimize transmit beamforming at the BS and waveform of the energy signal, we assume that the BS transmits a pulse-modulated RF energy signal employing a sequence of $N$ energy signal vectors $\boldsymbol{x}_n \in \mathbb{C}^{N_\text{t}}$, cf. Fig.~\ref{Fig:TimeSharing}.
The \gls*{ecb} representation of this pulse-modulated RF energy signal is modelled as $\boldsymbol{x}(t) = \sum_{n=1}^{N} \boldsymbol{x}_n \psi_n(t) \in \mathbb{C}^{N_\text{t}}$, where $\psi_n(t) = \Pi\big(\frac{ t - \sum_{k=0}^{n-1}\tau^\text{d}_{k} T_\text{f} }{ \tau^\text{d}_{n} T_\text{f} } \big)$ is the transmit pulse, $\tau_{0}^\text{d} = 0$,  $\Pi(t)$ is a rectangular function that takes value $1$ if $t \in [0,1)$ and $0$, otherwise.
Here, $\tau^\text{d}_{n}, n\in\{1,2,\cdots, N\},$ is the portion of the time frame of length $T_\text{f}$ utilized for transmission of energy signal vector $\boldsymbol{x}_n$ with $\sum_{n=1}^{N} \tau^\text{d}_{n} = \bar{\tau}$.

The \gls*{rf} energy signal received at user $k$ in the energy transfer phase is given by 
\begin{equation}
z^{\text{RF}}_{k}(t) = \sqrt{2} \Re \Big\{ \boldsymbol{h}_k  \boldsymbol{x}(t) \exp(j 2 \pi f_c t)  \Big\},
\end{equation}
\noindent where $f_c$ denotes the carrier frequency.
Similar to \cite{Morsi2019}, we neglect the noise received at the user terminals since its contribution to the harvested power is negligible.

To harvest power, users are equipped with memoryless non-linear \gls*{eh} circuits \cite{Shanin2021a}.
In contrast to the WPCN designs in \cite{Ju2014, Liu2014, Yang2015, Lee_2016, Morsi2018a, Boshkovska2015, Boshkovska2017a, Boshkovska2018, Hua2022, Zeng2021, Li2022}, in this work, we characterize the instantaneous behaviour of an EH circuit and model it via the relationship between the instantaneous received and harvested powers.
Furthermore, to take all non-linear and saturation effects of EH circuits into account, we assume that the mapping between the instantaneous received and harvested powers at user $k$ is characterized by a function $\phi_k(\cdot)$, which is given by
\begin{equation}
	\phi_k(|z|^2) = \min \{\varphi_k(|z|^2), \varphi_k(A_k^2)\}.
	\label{Eqn:EhModel}
\end{equation}
\noindent Here, $z$ is the \gls*{ecb} representative of the received RF energy signal at the user device, $\varphi_k(\cdot), k\in\{1,2,\cdots, K\},$ is a convex monotonically increasing function, whose parameters are independent of the received signal and determined by the EH circuit, and $A_k$ is the minimum amplitude of the received signal that drives the EH circuit of user $k$ into saturation \cite{Morsi2019, Shanin2021a}.
We note that function $\varphi_k(\cdot)$ can be, e.g., linear \cite{Ju2014, Liu2014, Yang2015, Lee_2016, Morsi2018a} or derived for a given rectifier circuit as in \cite{Clerckx2015, Morsi2019, Shanin2020}.
For our numerical simulations, we utilize the EH model whose parameters are summarized in Table~\ref{Table:SimulationSetup} in Section~\ref{Section:SimulationResults}.
 
Thus, the average power harvested by user $k$ during the energy transfer phase can be expressed as follows
\begin{equation}
	{p}^\text{d}_{k} =\sum_{n=1}^{N} \tau^\text{d}_{n} \phi_k\big(|z_{k,n}|^2\big),
	\label{Eqn:HarvestedPower}
\end{equation}
\noindent where $z_{k,n} = \boldsymbol{h}_k \boldsymbol{x}_n$ is the \gls*{ecb} representative of RF energy signal $z^{\text{RF}}_{k}(t)$ received in time slot $n$ at user $k$.
Furthermore, we assume that user $k$ is equipped with a rechargeable built-in battery, whose initial energy, $q_{k}$, is known\footnotemark\hspace*{0pt} at the BS \cite{Morsi2018a}.
Thus, at the end of the downlink transmission phase, the amount of energy available at user $k$ is given by 
\begin{equation}
	e_{k}(\boldsymbol{X}, \boldsymbol{\tau}^\text{d}, N) = q_{k} + T_\text{f} \sum_{n=1}^{N} \tau^\text{d}_{n} \phi_k\big(|\boldsymbol{h}_k \boldsymbol{x}_n|^2\big),
\end{equation}
\noindent\hspace*{0pt}where matrix $\boldsymbol{X} = [\boldsymbol{x}_1 \; \boldsymbol{x}_2\; \cdots\; \boldsymbol{x}_N]$ contains the transmit energy signal vectors and $\boldsymbol{\tau}^\text{d} = [\tau^\text{d}_1, \tau^\text{d}_2, \cdots, \tau^\text{d}_N]^\top$ is the vector of time durations allocated to the corresponding time slots.

\footnotetext{We note that if the initial energies of the user devices are not known at the BS, we can assume $q_k = \SI{0}{\joule}, \forall k$.}
\subsection{Information Transmission Phase}
In the uplink phase, the users transmit information to the BS exploiting the energy available in their batteries.
The symbol vector $\boldsymbol{r} \in \mathbb{C}^{N_\text{t}}$ received at the BS in the scheduled time slot is given by
\vspace*{-10pt}
\begin{equation}
	\vspace*{-5pt}
	\boldsymbol{r} = \sum_{k=1}^{K} \boldsymbol{h}_k^H \sqrt{p^\text{u}_{k}} s_k + \boldsymbol{n},
\end{equation}
where $p^\text{u}_{k}$ is the power utilized by user $k$ for transmission of information symbol $s_k \in \mathbb{C}$, which is modelled as a complex zero-mean and unit-variance Gaussian random variable, and $\boldsymbol{n} \in \mathbb{C}^{N_\text{t}}$ is an \gls*{awgn} vector with zero mean and covariance matrix $\sigma^2 \boldsymbol{I}_{N_\text{t}}$.
To reduce computational complexity and suppress inter-user interference, we assume \gls*{zf} equalization\footnotemark\hspace*{0pt} at the BS.
\footnotetext{We note that \gls*{zf} equalization is close to optimal if a large number of antennas is deployed at the BS, i.e., $N_\text{t} \gg K$, which is a preferred regime for wireless power transfer systems, where a significant downlink beamforming gain is generally required to be able to harvest meaningful amounts of power \cite{Yang2015, Interdonato2020}.}
Thus, the detected information symbol $\hat{s}_k$ of user $k$ at the BS can be expressed as follows:
\vspace*{-10pt}
\begin{equation}
	\vspace*{-7pt}
	\hat{s}_k = \boldsymbol{f}_{k} \boldsymbol{r} = \sqrt{p^\text{u}_{k}} s_k + \tilde{n}_k,
\end{equation}
\noindent where $\tilde{n}_k = \boldsymbol{f}_{k} \boldsymbol{n}$ is the equivalent \gls*{awgn} with variance $\tilde{\sigma}^2_k = \|\boldsymbol{f}_{k}\|^2_2\sigma^2$ impairing the detected information signal transmitted by user $k$.
Here, equalization vector $\boldsymbol{f}_{k} \in \mathbb{C}^{1 \times N_\text{t}}$ is the $k^\text{th}$ row of matrix $\boldsymbol{F} = (\boldsymbol{H}_\text{u}^H \boldsymbol{H}_\text{u})^{-1} \boldsymbol{H}_\text{u}^H$, where $\boldsymbol{H}_\text{u} = [\boldsymbol{h}_1^H \; \boldsymbol{h}_2^H \; \cdots \; \boldsymbol{h}_K^H]$ is the composite uplink channel between the users and the BS.
Finally, the data rate of user $k$ is given by $ R_k(\bar{\tau}, {p}_k^\text{u}) = (1-\bar{\tau}) \log_2(1 + \Gamma_{k})$, where $\Gamma_{k} = { p^\text{u}_{k}}/{\tilde{\sigma}_k^2 }$ is the \gls*{snr} at the BS for the detected information symbol transmitted by user $k$.

\section{Problem Formulation and Optimal Solution}
\label{Section:ProblemFormulationAndOptimalSolution}
In this section, we jointly optimize the waveform of the downlink energy signal, i.e., the downlink transmit energy signal vectors, their number, the durations of the transmit pulses, and the time allocation policy, i.e., the normalized durations of the downlink and uplink subframes.
To this end, we first formulate an optimization problem, and then, we determine the optimal solution for single-user and multi-user WPCNs, respectively.
\subsection{Problem Formulation}
In the following, we formulate an optimization problem for the minimization of the average transmit power at the BS under per-user rate constraints in the uplink.
For a given time frame, the optimal transmit energy signal vectors and time allocation policy are obtained as the solution of the following non-convex optimization problem:
\begin{subequations}
	\vspace*{-10pt}
	\begin{align}
		\minimize_{\substack{\boldsymbol{\tau}^\text{d} \succeq \boldsymbol{0}, \bar{\tau} \in [0,1],  \boldsymbol{X}, \boldsymbol{p}^\text{u}, N}  } \; &P_{\text{DL}} (\boldsymbol{\tau}^\text{d}, \boldsymbol{X}, N) \label{Eqn:ProblemObj} \\
		\subjectto \; \quad & R_k(\bar{\tau}, {p}_k^\text{u}) \geq R^{\text{req}}_{k}, \, \forall k, \label{Eqn:ProblemC1} \\
		& e_{k} (\boldsymbol{X}, \boldsymbol{\tau}^\text{d}, N) \geq (1-\bar{\tau})p^\text{u}_{k} T_\text{f} + p^\text{req}_kT_\text{f}, \, \forall k, \label{Eqn:ProblemC2} \\
		&\sum_{n=1}^{N} \tau^\text{d}_{n} = \bar{\tau}, \label{Eqn:ProblemC3}
	\end{align}
	\label{Eqn:OriginalProblem}
\end{subequations}
\noindent\hspace*{-3pt}where $P_{\text{DL}}(\boldsymbol{\tau}^\text{d}, \boldsymbol{X}, N) = \sum_{n=1}^{N} \tau^\text{d} _n \| \boldsymbol{x}_n \|_2^2$ and $\boldsymbol{p}^\text{u} = [p^\text{u}_1 \; p^\text{u}_2 \; \cdots \; p^\text{u}_K]^\top$ are the average transmit power in the energy transfer phase and the vector of powers utilized for information transmission in the uplink, respectively.
In (\ref{Eqn:OriginalProblem}), $R^{\text{req}}_{k}$ and $p^\text{req}_k$ are the required rate in the uplink and the power needed at user $k$ to perform operational tasks other than information transmission, e.g., sensing and signal processing tasks, respectively.
We assume that the values of $R^{\text{req}}_{k}$ and $p^\text{req}_k$ are set by user $k$ and communicated to the BS beforehand. 
We note that in contrast to \cite{Ju2014, Liu2014, Yang2015, Lee_2016, Morsi2018a, Boshkovska2015, Boshkovska2017a, Boshkovska2018, Hua2022, Zeng2021, Li2022, Clerckx2015, Clerckx2017, Zawawi2019}, where the covariance matrix $\tilde{\boldsymbol{X}} = \sum_{n=1}^{N} \frac{\tau^\text{d}_n}{T_\text{f}} \boldsymbol{x}_n \boldsymbol{x}_n^H$ of the downlink energy signal is optimized, the EH model in (\ref{Eqn:EhModel}) enables the direct optimization of the individual transmit energy signal vectors $\boldsymbol{x}_1, \boldsymbol{x}_2, \cdots, \boldsymbol{x}_N$ of the pulse-modulated energy signal $\boldsymbol{x}(t)$.

Due to the non-convexity of the objective function (\ref{Eqn:ProblemObj}) and constraints (\ref{Eqn:ProblemC1}) and (\ref{Eqn:ProblemC2}), optimization problem (\ref{Eqn:OriginalProblem}) is, in general, difficult to solve.
In the following, we first show that for single-user WPCNs, optimization problem (\ref{Eqn:OriginalProblem}) can be solved in closed-form.
Next, for multi-user WPCNs, we characterize the optimal solution of (\ref{Eqn:OriginalProblem}) for the general case and determine it based on a $K$-{dimensional} grid search.
\subsection{Single-User WPCNs}
\label{Section:SingleUser}
In this section, we solve optimization problem (\ref{Eqn:OriginalProblem}) for single-user WPCNs and thus, for notational convenience, we drop the subscripts for user $1$.
The optimal power ${p^\text{u}}^{*}$ for uplink information transmission has to satisfy constraints (\ref{Eqn:ProblemC1}) and (\ref{Eqn:ProblemC2}) and thus, can be chosen arbitrarily\footnotemark\hspace*{0pt} in the interval $[p^\text{u}_\text{min}, p^\text{u}_\text{max}]$.
Here, $p^\text{u}_\text{min} = (2^{\frac{R^\text{req}}{1-\bar{\tau}}}-1) \tilde{\sigma}^2$ is the minimum power required for uplink transmission with required rate $R^\text{req}$ and $p^\text{u}_\text{max} = \frac{1}{1-\bar{\tau}} \big( \frac{q}{T_\text{f}} - p^\text{req} + \sum_n {\tau^\text{d}_n} \phi(|\boldsymbol{h} \boldsymbol{x}_n|^2) \big)$ is the power available at the user device at the end of the energy transfer phase.
The optimal time sharing parameter $\bar{\tau}^*$, the normalized lengths of the time slots ${\tau^\text{d}_n}^*$, and the transmit energy signal vectors in the downlink $\boldsymbol{x}_n^*, n\in\{1,2,\cdots, N\},$ can be obtained as the solution of the following optimization problem:
\vspace*{-5pt}
\begin{equation}
		\minimize_{\boldsymbol{\tau}^\text{d} \succeq \boldsymbol{0}, \substack{\bar{\tau} \in [0,1], \boldsymbol{X}, N}  } \; \sum_{n=1}^{N} \tau^\text{d}_n \| \boldsymbol{x}_n \|_2^2 \quad
		\subjectto \;  \sum_{n=1}^{N} \tau^\text{d}_{n} \phi(|\boldsymbol{h} \boldsymbol{x}_n|^2) \geq \xi(\bar{\tau}), \;
		\sum_{n=1}^{N} \tau^\text{d}_{n} = \bar{\tau},
	\label{Eqn:SUProblem}
	\vspace*{-5pt}
\end{equation}
\noindent\hspace*{-0pt}where $\xi(\bar{\tau}) = p^\text{req} - \frac{q}{T_\text{f}} + (1-\bar{\tau}) \big( 2^{\frac{R^\text{req}}{1-\bar{\tau}}} - 1 \big) \tilde{\sigma}^2$.
\footnotetext{We note that for power-efficient communication, if $p^\text{u}_\text{max} > p^\text{u}_\text{min}$, the user devices may select ${p^\text{u}}^{*} = p^\text{u}_\text{min}$ to save power for future use.}
We note that similar to (\ref{Eqn:OriginalProblem}), problem (\ref{Eqn:SUProblem}) is still non-convex.
However, in the following proposition, we show that the optimal solution of (\ref{Eqn:SUProblem}) requires only a single energy signal vector, i.e., $N^*=1$, which is collinear with the \gls*{mrt} vector.
\begin{proposition}
	The optimal transmit signal in the energy transfer phase as solution of problem (\ref{Eqn:SUProblem}) employs a single energy signal vector $\boldsymbol{x}^*~=~ \boldsymbol{w}^* s^*$, where $\boldsymbol{w}^* = \frac{\boldsymbol{h}^H}{\| \boldsymbol{h} \|_2}$ is the \gls*{mrt} beamformer and $s^*= \alpha^*_s \exp(j\theta_s)$ is a scalar symbol.
	The magnitude of optimal symbol $s^*$ is given by $\alpha_s^* = \frac{A}{\| \boldsymbol{h} \|_2}$, whereas the phase of $s^*$ can be arbitrarily chosen, i.e., $\theta_s \in [0, 2\pi)$.
	\label{Theorem:SingleUser}
\end{proposition}
\begin{proof}
	Please refer to Appendix~\ref{Appendix:PropSU}.	
\end{proof}

Proposition~\ref{Theorem:SingleUser} implies that in the downlink energy transfer phase, the optimal signal consists of a single transmit energy signal vector $\boldsymbol{x}^*$, which is collinear with the MRT beamformer and whose magnitude is chosen to drive the EH circuit into saturation.
Exploiting Proposition~\ref{Theorem:SingleUser}, we determine the optimal $\bar{\tau}^*$ that solves (\ref{Eqn:SUProblem}) as follows:
\vspace*{-5pt}
\begin{equation}
	\vspace*{-5pt}
	\bar{\tau}^* = \min \{\bar{\tau}: f_\text{SU}(\bar{\tau}) \geq 0\},
	\label{Eqn:SingleUserProblemTau}
\end{equation}
\noindent where $f_\text{SU}(\bar{\tau})  = \bar{\tau} \phi( A^2\big) - \xi(\bar{\tau})$. 
Furthermore, if problem (\ref{Eqn:SUProblem}) is feasible and does not have a non-trivial solution, i.e., $\bar{\tau}^* > 0$ or, equivalently, $f_\text{SU}(0) = \frac{q}{T_\text{f}} - p^\text{req} - \big( 2^{R^\text{req}} - 1 \big) \tilde{\sigma}^2 < 0$, the optimal $\bar{\tau}^*$ can be obtained as the minimum root of equation $f_\text{SU}(\bar{\tau}) = 0$, i.e., $\bar{\tau}^* = \min \{\bar{\tau}: f_\text{SU}(\bar{\tau}) = 0\}$.

\subsection{Multi-User WPCNs}
In this section, we determine the optimal solution of (\ref{Eqn:OriginalProblem}) for multi-user WPCNs, i.e., for $K > 1$.
First, we reformulate constraint (\ref{Eqn:ProblemC1}) equivalently as $p^\text{u}_k \geq \xi^\text{u}_k (\bar{\tau}), \, \forall k,$ where $\xi^\text{u}_k(\bar{\tau}) = \tilde{\sigma}_k^2 \big( 2^\frac{R^\text{req}_k}{ 1 - \bar{\tau} } - 1 \big)$ is the minimum power required by user $k$ for information transmission with rate $R_k^\text{req}$. 
Next, similarly to single-user WPCNs in Section~\ref{Section:SingleUser}, for given $\bar{{\tau}}$, $\boldsymbol{\tau}^\text{d}, \boldsymbol{X}$, and $N$, the optimal uplink transmit power $p^\text{u*}_k$ for user $k$ satisfying constraints (\ref{Eqn:ProblemC1}) and (\ref{Eqn:ProblemC2}) can be arbitrarily chosen from the interval $[\xi^\text{u}_k(\bar{\tau}), p^\text{u}_{\text{max}, k}]$, where $p^\text{u}_{\text{max}, k}$ is the maximum power available at user $k$ for uplink transmission and is given by
\vspace*{-5pt}
\begin{equation}
	\vspace*{-5pt}
	p^\text{u}_{\text{max}, k} = \frac{1}{1-\bar{\tau}} \Big( \frac{q_k}{T_\text{f}} - p_k^\text{req} + \sum_{n=1}^{N} {\tau^\text{d}_n} \phi_k(|\boldsymbol{h}_k \boldsymbol{x}_n|^2) \Big).
\end{equation}
Then, problem (\ref{Eqn:OriginalProblem}) can be simplified as follows:
\begin{equation}
	\minimize_{\boldsymbol{\tau}^\text{d} \succeq \boldsymbol{0}, \bar{\tau} \in [0,1], \boldsymbol{X}, N } \; \sum_{n=1}^{N} \tau^\text{d}_{n} \| \boldsymbol{x}_n \|_2^2 \quad
	\subjectto \; \sum_{n=1}^{N} \tau^\text{d}_{n} = \bar{\tau}, \; \sum_{n=1}^{N} \tau^\text{d}_{n} \phi_k\big(|\boldsymbol{h}_k \boldsymbol{x}_n|^2\big) \geq \xi^\text{d}_k(\bar{\tau}), \, \forall k,
	\label{Eqn:RefProblem}
\end{equation}
\noindent where $\xi^\text{d}_k(\bar{\tau}) = (1-\bar{\tau})\xi^\text{u}_{k}(\bar{\tau}) + p^\text{req}_k - \frac{q_k}{T_\text{f}}$ is the required harvested power at user $k$.
	
In the following proposition, we characterize the optimal solution of (\ref{Eqn:RefProblem}).
\begin{proposition}
	The optimal transmit energy signal vectors $\boldsymbol{x}^*_n, n\in\{1,2,\cdots, N\},$ as solution of (\ref{Eqn:RefProblem}) can be expressed as follows
	\vspace*{-10pt}
	\begin{equation}
		\vspace*{-5pt}
		\boldsymbol{x}^*_n = \boldsymbol{w}_n^* s_n,
		\label{Eqn:TransmitSymbolDecomposition}
	\end{equation}
	\noindent where $s_n = \exp(j \theta_n), n \in \{1,2,\cdots, N\},$ are scalar unit-norm symbols with arbitrary phases $\theta_n = [0, 2\pi)$. 
	Energy beamforming vectors $\boldsymbol{w}_n^*$, $n\in\{1,2,\cdots, N\}$, can be obtained as solutions of the following optimization problem:
	\vspace*{-5pt}
	\begin{equation}
		\vspace*{-5pt}
		\minimize_{\boldsymbol{w}_n} \; \|\boldsymbol{w}_n\|_2^2 \quad \subjectto \; \phi_k(|\boldsymbol{h}_k \boldsymbol{w}_n|^2) \geq \mu_{n,k}, \, \forall k.
		\label{Eqn:OptimalTransmitVectors}
	\end{equation}
	For each $n\in\{1,2,\cdots, N\}$, vector $\boldsymbol{\mu}_n = [\mu_{n,1}, \mu_{n,2}, \cdots, \mu_{n,K}]^\top \in \mathbb{R}_{+}^K$ contains the instantaneous powers harvested at user devices $1,2,\cdots, K$ in time slot $n$.
	Furthermore, the optimal $\bar{\tau}^*$, ${\boldsymbol{\tau}^\text{\upshape{d}} }^*$, vectors $\boldsymbol{\mu}^*_n, n\in\{1,2,\cdots, N\}$, and $N^*$ can be determined as solution of the following resource allocation problem:
	\vspace*{-10pt}
	\begin{equation}
		\vspace*{-5pt}
		 	\minimize_{\boldsymbol{\tau}^\text{\upshape{d}} \succeq \boldsymbol{0}, \bar{\tau} \in [0,1], \boldsymbol{\mu}_1, \cdots, \boldsymbol{\mu}_{N}, N} \; \sum_{n=1}^{N} \tau^\text{\upshape{d}}_n \psi(\boldsymbol{\mu}_n) \quad
			\subjectto \; \sum_{n=1}^{N} \tau^\text{\upshape{d}}_{n} \boldsymbol{\mu}_{n} \succeq \boldsymbol{\xi}^\text{\upshape{d}}(\bar{\tau}), \;
			\sum_{n=1}^{N} \tau^\text{\upshape{d}}_{n} = \bar{\tau},
	\label{Eqn:OptimalResourceAlloc}
	\end{equation}
	\noindent\hspace*{0pt}with $\boldsymbol{\xi}^\text{\upshape{d}}(\bar{\tau}) = [{\xi}^\text{\upshape{d}}_1(\bar{\tau}), \; {\xi}^\text{\upshape{d}}_2(\bar{\tau}), \; \cdots, \; {\xi}^\text{\upshape{d}}_K(\bar{\tau})]^\top$.
	Here, $\psi(\boldsymbol{\mu}) : \mathcal{R}^K \to \mathcal{R}$ is a monotonic non-decreasing function given by
	\vspace*{-10pt}
	\begin{equation}
		\vspace*{-5pt}
		\psi(\boldsymbol{\mu}) = \min_{\boldsymbol{x} \in \Omega(\boldsymbol{\mu})} \| \boldsymbol{x} \|^2_2
		\label{Eqn:FunctionPsi}
	\end{equation} 
	\noindent with $\Omega(\boldsymbol{\mu}) = \{\boldsymbol{x}: \phi_k(|\boldsymbol{h}_k \boldsymbol{x}|^2) \geq \mu_{k}, \forall k\}$.
	\label{Theorem:MuProp1}
\end{proposition}
\begin{proof}
	Please refer to Appendix \ref{Appendix:Prop1}.
\end{proof}

Proposition~\ref{Theorem:MuProp1} reveals that problem (\ref{Eqn:RefProblem}) can be solved by first determining function $\psi(\boldsymbol{\mu})$ in (\ref{Eqn:FunctionPsi}), and subsequently obtaining the optimal time sharing parameter, number of time slots, their normalized lengths, and vectors of harvested powers in the downlink, $\bar{\tau}^*, N^*, {\boldsymbol{\tau}^\text{d}}^*$, and $\boldsymbol{\mu}^*_n, n\in\{1,2,\cdots, N^*\}$, respectively, as solution of optimization problem (\ref{Eqn:OptimalResourceAlloc}).
Finally, the optimal downlink transmit energy signal vectors in (\ref{Eqn:TransmitSymbolDecomposition}) comprise energy beamformers $\boldsymbol{w}_n^*$ and scalar unit-norm symbols\footnotemark\hspace*{0pt} $s_n$.
The optimal energy beamforming vectors $\boldsymbol{w}_n^*$ in (\ref{Eqn:TransmitSymbolDecomposition}) can be obtained as solutions of the non-convex problem (\ref{Eqn:OptimalTransmitVectors}).
\footnotetext{Since scalar phase  $\theta_n$ of $s_n^*$ can be chosen arbitrarily, it can be utilized for, e.g., information transmission in the downlink \cite{Shanin2021a}.}

In the following, we separately solve optimization problems (\ref{Eqn:OptimalTransmitVectors}), (\ref{Eqn:OptimalResourceAlloc}), and (\ref{Eqn:FunctionPsi}).
Although problems (\ref{Eqn:OptimalTransmitVectors}) - (\ref{Eqn:FunctionPsi}) are non-convex, we show that the computational complexity of solving (\ref{Eqn:OptimalResourceAlloc}) and (\ref{Eqn:FunctionPsi}) depends on the number of users $K$ and is independent of the number of BS antennas $N_\text{t}$, whereas the solution of (\ref{Eqn:OptimalTransmitVectors}) has a computational complexity that is polynomial in $N_\text{t}$.

	\subsubsection{Solution of Problem (\ref{Eqn:FunctionPsi})}
	\label{Section:OptimalFunctionPsi}
	In this section, for any given vector $\boldsymbol{\mu} \in \mathcal{R}^K$, we determine the value of function $\psi(\boldsymbol{\mu})$ in (\ref{Eqn:FunctionPsi}).
To this end, in the following proposition, we show that problem (\ref{Eqn:FunctionPsi}) can be equivalently reformulated as a convex optimization problem.
\begin{proposition}
	For a given vector $\boldsymbol{\mu} \in \mathcal{R}^K$, the value of function $\psi(\boldsymbol{\mu})$ can be determined as solution of the following convex optimization problem:
	\vspace*{-7pt}
	\begin{equation}
		\vspace*{-7pt}
		\psi(\boldsymbol{\mu}) = \max_{\boldsymbol{\lambda} \succeq \boldsymbol{0}} \; \boldsymbol{\rho}^\top \boldsymbol{\lambda} \quad \subjectto \; \| \boldsymbol{B} \diag(\boldsymbol{\lambda}) \boldsymbol{B} \|_2 \leq 1,
		\label{Eqn:FunctionPsiProposition}
	\end{equation}
	\noindent where $\boldsymbol{\lambda} \in \mathbb{R}^K$, $\rho_k = \phi_k^{-1}(\mu_k), k\in\{1,2,\cdots,K\}$, and $\boldsymbol{B} = (\boldsymbol{H} \boldsymbol{H}^H)^{\frac{1}{2}}$ with $\boldsymbol{H} = \boldsymbol{H}_\text{u}^H = [\boldsymbol{h}_1^H \; \boldsymbol{h}_2^H \; \cdots \; \boldsymbol{h}_K^H]^H$.
	\label{Theorem:MuProp2}
\end{proposition}
\begin{proof}
	Please refer to Appendix~\ref{Appendix:Prop2}.
\end{proof}

Proposition~\ref{Theorem:MuProp2} reveals that, for any given $\boldsymbol{\mu}$, the value of $\psi(\boldsymbol{\mu})$ can be obtained as a solution of the convex optimization problem (\ref{Eqn:FunctionPsiProposition}) via a standard numerical optimization tool, such as CVX \cite{Grant2015}. 
Furthermore, we highlight that the computational complexity of (\ref{Eqn:FunctionPsiProposition}) does not depend on the number of transmit antennas $N_\text{t}$ \cite{Polik2010}. 

	\subsubsection{Solution of Problem (\ref{Eqn:OptimalResourceAlloc})}
	\label{Section:OptimalResourceAllocation}
	In the following, for function $\psi(\boldsymbol{\mu})$ introduced in (\ref{Eqn:FunctionPsi}), we solve optimization problem (\ref{Eqn:OptimalResourceAlloc}).
To this end, in the following proposition, we first show that for any function $\psi(\boldsymbol{\mu})$, the optimal value of $N$ satisfies $N^* \leq K+1$.
Next, we propose a grid search algorithm to solve (\ref{Eqn:OptimalResourceAlloc}).

\begin{proposition}
	For any given function $\psi(\boldsymbol{\mu})$, the optimal number of time slots for the downlink power transfer phase satisfies $N^*\leq K+1$.
	\label{Theorem:MuProp3}
\end{proposition}
\begin{proof}
	Please refer to Appendix \ref{Appendix:Prop3}.
\end{proof}

Proposition~\ref{Theorem:MuProp3} provides an upper bound on the optimal number of time slots $N^*$ in the downlink phase.
Therefore, to find the optimal solution of (\ref{Eqn:OptimalResourceAlloc}), we set $N =\bar{N} = K + 1$ and determine the optimal $\bar{\tau}^* \in [0,1]$, ${\boldsymbol{\tau}^\text{d}}^* \in [0,1]^{\bar{N}}$, and $\boldsymbol{\mu}_n^*, n \in \{1,2,\cdots, \bar{N}\}$, as solution of the resulting optimization problem.

In the following, we exploit a grid search to solve (\ref{Eqn:OptimalResourceAlloc}).
First, we define a uniform grid $\mathcal{P}_\tau$ = $\{\bar{\tau}_1, \bar{\tau}_2, \cdots \bar{\tau}_{L_\tau}\}$ of size $L_\tau$, where $\bar{\tau}_p = \frac{p-1}{L_\tau - 1}, p\in\{1,2,\cdots L_\tau - 1\}$.
Next, we obtain the optimal ${\boldsymbol{\tau}^\text{d}}^*$ and $\boldsymbol{\mu}_n^*, n \in \{1,2,\cdots, \bar{N}\},$ for each $\bar{\tau} \in \mathcal{P}_\tau$, i.e., we solve the following optimization problem:
\begin{equation}
		\minimize_{\boldsymbol{\tau}^\text{d} \succeq \boldsymbol{0}, \boldsymbol{\mu}_1, \cdots, \boldsymbol{\mu}_{\bar{N}}} \; \sum_{n=1}^{\bar{N}} \tau^\text{d}_n \psi(\boldsymbol{\mu}_n) \quad
		\subjectto \; \sum_{n=1}^{\bar{N}} \tau^\text{d}_{n} {\mu}_{n,k} \geq {\xi}_k^\text{d}(\bar{\tau}), \; \sum_{n=1}^{\bar{N}} \tau^\text{d}_{n} = \bar{\tau}.
	\label{Eqn:RefOptimalResourceAlloc}
\end{equation}
To this end, we define a uniform grid $\mathcal{P}_\mu$ of size $L_\mu^K$.
Each element of $\mathcal{P}_\mu$ is a vector $\boldsymbol{\mu}_j = [{\mu}_{j,1}, {\mu}_{j,2}, \cdots, {\mu}_{j,K}]^\top$, $j~\in~\{1,2,\cdots, L_\mu^K\}$, that represents a unique combination of harvested powers at the user devices, i.e., $\boldsymbol{\mu}_i \neq \boldsymbol{\mu}_j, \forall i\neq j$.
We note that the harvested powers $\mu_{n,k}, n\in\{1,2,\cdots,\bar{N}\}, k\in\{1,2,\cdots, K\}$, are bounded, i.e., $\mu_{n,k} \in [0, \varphi_k(A_k^2)], \forall n, k$.
Therefore, to uniquely define vector $\boldsymbol{\mu}_j$, we set its elements to the values ${\mu}_{j,k} = \frac{i_k \phi_k(A_k^2)}{L_\mu - 1}$, where $i_k$ is the digit in position $k$ of the representation of $j-1$ in the numeral system with base $L_\mu$, $j\in\{1,2,\cdots L_\mu^K\}, k\in\{1,2,\cdots, K\}$, i.e., $i_k \in \{0,1,\cdots, L_\mu-1\}$ and $j = \sum_{k=1}^K i_k L_\mu^{K-k} + 1$.

Finally, we solve (\ref{Eqn:RefOptimalResourceAlloc}) on the grid $\mathcal{P}_\mu$, i.e., we optimally select the $K+1$ elements of $\mathcal{P}_\mu$ with non-zero\footnotemark\hspace*{0pt} normalized time durations.
\footnotetext{In our simulations, since the number of such elements may be smaller than $K+1$, we select $K+1$ elements of $\mathcal{P}_\mu$ which correspond to the largest normalized time durations.}
To this end, we solve the following linear optimization problem:
\vspace*{-7pt}
\begin{equation}
	\vspace*{-7pt}
	 \minimize_{\boldsymbol{\tau}^\text{d}\succeq \boldsymbol{0}} \; \sum_{j=1}^{L_\mu^K} \tau^\text{d}_j \psi(\boldsymbol{\mu}_j) \quad \subjectto \; \boldsymbol{M} \boldsymbol{\tau}^\text{d} \succeq \boldsymbol{\xi}^\text{d}(\bar{\tau}),\; \sum_{j=1}^{L_\mu^K} {\tau_j}^\text{d} = \bar{\tau},
	\label{Eqn:LinearProblemTau}
\end{equation}
\noindent where $\boldsymbol{M} = [\boldsymbol{\mu}_1, \boldsymbol{\mu}_2, \cdots, \boldsymbol{\mu}_{L_\mu^K}] \in \mathbb{R}^{K\times L_\mu^K}$.
Since problem (\ref{Eqn:LinearProblemTau}) is linear, it can be efficiently solved using a standard numerical optimization tool, such as CVX \cite{Grant2015}.
Finally, we choose $\bar{{\tau}}^* \in \mathcal{P}_\tau$ that yields the minimum transmit power in the downlink.
For the optimal $\bar{{\tau}}^*$, the optimal normalized lengths of the time slots ${\boldsymbol{\tau}}^\text{d*}$ and the corresponding vectors of harvested powers $\boldsymbol{\mu}_n^*$, $n~\in~\{1,2,\cdots, \bar{N}\},$ are obtained as the $\bar{N}$ non-zero elements of $\boldsymbol{\tilde{\tau}}^\text{d*}(\bar{{\tau}}^*)$ and the corresponding elements of $\mathcal{P}_\mu$, respectively, where $\boldsymbol{\tilde{{\tau}}}^\text{d*}(\bar{{\tau}})$ is the solution of (\ref{Eqn:LinearProblemTau}) for a given $\bar{{\tau}}$.

The algorithm for determining the optimal solution of problem (\ref{Eqn:OptimalResourceAlloc}) is summarized in Algorithm~1.
We note that the computational complexity of the algorithm is exponential in the number of users $K$ but does not depend on the number of transmit antennas $N_\text{t}$ employed at the BS.

\begin{algorithm}[!t]		
	\small				
	\linespread{1.45}\selectfont
	\SetAlgoNoLine%
	\SetKwFor{Foreach}{for each}{do}{end}		
	Initialize: Grid sizes $L_{\mu}$ and $L_\tau$, channel covariance matrix $\boldsymbol{B} = (\boldsymbol{H} \boldsymbol{H}^H)^{\frac{1}{2}}$, required rate $R_k^\text{req}$, required power $p_k^\text{req}$, and initial energy $q_k$ at user $k$ for $k \in \{1,2,\cdots, K\}$.	\\	
	1. Create grid $\mathcal{P}_\mu$ and calculate $\psi(\cdot)$ on its vertices: \\
	\For{$j = 0$ {\upshape to} $L_{\mu}^K-1$}{
		1.1. Create vector $\boldsymbol{\mu}_j = [\mu_{j,1} \; \mu_{j,2}\; \cdots \; \mu_{j,K}]^\top$, where $\mu_{j,k} = \frac{i_k \, \phi_k(A_k^2)}{L_\mu - 1}$, $i_k \in \{0,1,\cdots, L_\mu\}$, and $j = \sum_{k=1}^K i_k L_\mu^{K-k}$ \\
		1.2. Calculate $\psi_{j} = \psi(\boldsymbol{\mu}_j)$ as solution of (\ref{Eqn:FunctionPsiProposition}) \\
	}
	2. Create grid $\mathcal{P}_\tau$ and solve (\ref{Eqn:LinearProblemTau}) for each point on the grid:\\
	\For{$p = 1$ {\upshape to} $L_{\tau}$}{
		2.1. Calculate the time sharing parameter $\bar{\tau}_p = \frac{p-1}{L_\tau - 1}$\\
		2.2. For $\bar{\tau} = \bar{\tau}_p$, determine ${\tilde{\boldsymbol{\tau}}^\text{d*}_p}$ as solution of (\ref{Eqn:LinearProblemTau}) and the corresponding value $\Psi^*_p = \boldsymbol{\psi}^\top {\tilde{\boldsymbol{\tau}}^\text{d*}_p}$  \\
	}
	3. Find $p^* = \argmin_p \Psi^*_p$, the ratio $\bar{\tau}^* = \bar{\tau}_{p^*}$, and the corresponding vectors ${\boldsymbol{\tau}^\text{d}}^*$ and $\boldsymbol{\mu}^*_n, n\in\{1,2,\cdots, \bar{N}\}$  \\
	\textbf{Output:} Optimal parameter $\bar{\tau}^*$, time slot lengths ${\boldsymbol{\tau}^\text{d}}^*$, and harvested powers $\boldsymbol{\mu}_n, n\in\{1,2,\cdots,\bar{N}\}$
	\caption{\strut Algorithm for determining the optimal solution of (\ref{Eqn:OptimalResourceAlloc}) }
	\label{Alg:OptimalResourceAllocation}
\end{algorithm}	 
	\subsubsection{Solution of Problem (\ref{Eqn:OptimalTransmitVectors})}
	\label{Section:OptimalVectors}
	In the following, for the optimal vectors $\boldsymbol{\mu}^*$ obtained in Section~\ref{Section:OptimalResourceAllocation}, we determine the optimal energy beamforming vectors $\boldsymbol{w}_n^*, n\in\{1,2,\cdots, \bar{N}\},$ as a solution of problem (\ref{Eqn:OptimalTransmitVectors}).
Due to the non-concavity of $\phi_k(\cdot), k\in\{1,2,\cdots, K\}$, problem (\ref{Eqn:OptimalTransmitVectors}) is non-convex.
Nevertheless, in the following, we show that problem (\ref{Eqn:OptimalTransmitVectors}) can be equivalently formulated as a convex optimization problem, whose solution has a computational complexity, which is polynomial in $N_\text{t}$.

Since problem (\ref{Eqn:OptimalTransmitVectors}) can be reformulated as a rank-constrained semidefinite optimization problem \cite{Shanin2021a}, in the following lemma, we first consider a class of convex semidefinite optimization problems with linear constraints and show that the solutions of such problems are low-rank.
The result of this lemma is then exploited for the solution of problem (\ref{Eqn:OptimalTransmitVectors}).
\begin{lemma}
	For any given vector $\boldsymbol{b} \in \mathbb{R}^K$, the optimal solution of the following convex optimization problem
	\vspace*{-7pt}
	\begin{equation}
		\vspace*{-5pt}
		\minimize_{\boldsymbol{X} \in \mathcal{S}^{N_\text{t}}_{+}} \; \text{\upshape{Tr}}\{\boldsymbol{X}\} \qquad
		\subjectto \; \boldsymbol{h}_k \boldsymbol{X} \boldsymbol{h}_k^H \geq b_k, \; \forall k\in\{1,2,\cdots,K\},
		\label{Eqn:LemmaProblem}
	\end{equation} 
	\noindent\hspace*{0pt}satisfies\footnotemark\hspace*{0pt} $\rank\{\boldsymbol{X}\}\leq~1$.
	\label{Theorem:Lemma}
\end{lemma}
\begin{proof}
	Please refer to Appendix~\ref{Appendix:LemmaProof}.
\end{proof}
\footnotetext{We note that if $b_k = 0, \forall k,$ the optimal solution satisfies $\rank\{\boldsymbol{X}\} = 0$. Otherwise, we have $\rank\{\boldsymbol{X}\} = 1$.}

Next, in the following proposition, we equivalently formulate (\ref{Eqn:OptimalTransmitVectors}) as a convex optimization problem.
\begin{proposition}
	\label{Theorem:MuProp4}
	The optimal energy beamforming vectors $\boldsymbol{w}_n^*, n\in\{1,2,\cdots, \bar{N}\}$, are given by $\boldsymbol{w}_n^* = \gamma_n \boldsymbol{v}_n$, where $\gamma_n$ and $\boldsymbol{v}_n$ are the dominant eigenvalue and the corresponding eigenvector of matrix $\boldsymbol{W}_n^*$, which is obtained as solution of the following convex optimization problem:
		\begin{equation}
				\minimize_{\boldsymbol{W}_n \in \mathcal{S}^{N_\text{t}}_{+}} \quad \text{\upshape{Tr}}\{\boldsymbol{W}_n\} \quad
				\subjectto \quad \boldsymbol{h}_k \boldsymbol{W}_n \boldsymbol{h}_k^H \geq \phi_k^{-1}(\mu^*_{n,k}), \; \forall k.
		\label{Eqn:OptimalVectorsRef}
		\end{equation} 
	\vspace*{-35pt}
\end{proposition}
\begin{proof}
	Since function $\phi_k(\cdot)$ is monotonically non-decreasing, problem (\ref{Eqn:OptimalTransmitVectors}) can be equivalently reformulated as follows
	\begin{equation}
		\minimize_{\boldsymbol{W}_n \in \mathcal{S}^{N_\text{t}}_{+}} \;  \text{\upshape{Tr}}\{\boldsymbol{W}_n\} \quad
		\subjectto \; \boldsymbol{h}_k \boldsymbol{W}_n \boldsymbol{h}_k^H \geq \phi^{-1}(\mu^*_{n,k}), \; \forall k, \quad
		\rank\{  \boldsymbol{W}_n \} \leq 1.
		\label{Eqn:ProofProp4}
	\end{equation} 
   	 Problems (\ref{Eqn:OptimalVectorsRef}) and (\ref{Eqn:ProofProp4}) are equivalent thanks to Lemma~\ref{Theorem:Lemma}.
	 Then, the optimal energy beamforming vector $\boldsymbol{w}_n^*, n\in\{1,2,\cdots, \bar{N}\},$ as solution of (\ref{Eqn:OptimalTransmitVectors}) can be determined as $\boldsymbol{w}_n^* = \gamma_n \boldsymbol{v}_n$, where $\gamma_n$ and $\boldsymbol{v}_n$ are the dominant eigenvalue and the corresponding eigenvector of matrix $\boldsymbol{W}_n^*$, solution of (\ref{Eqn:ProofProp4}).
	 This concludes the proof.
\end{proof}

Proposition~\ref{Theorem:MuProp4} reveals that the optimal energy beamforming vectors $\boldsymbol{w}_n^*, n\in\{1,2,\cdots, \bar{N}\},$ can be determined as solution of convex optimization problem (\ref{Eqn:OptimalVectorsRef}).
We note that (\ref{Eqn:OptimalVectorsRef}) can be efficiently solved using numerical optimization tools, such as CVX \cite{Grant2015}.

The algorithm for determining the optimal solution of problem (\ref{Eqn:OriginalProblem}) is summarized in Algorithm~\ref{Alg:OptimalSolution}.
The computational complexity of the optimal WPCN design as function of $N_\text{t}$ and $K$ is given by $\Theta_{\text{Opt}}(N_\text{t}, K) =  \Theta_{\text{RA}}(N_\text{t}, K) + \Theta_{\text{BF}}(N_\text{t}, K)$.
Here, $\Theta_{\text{RA}}(N_\text{t}, K) = \mathcal{O}\big( L_{\mu}^K K^3 \big)$ and $\Theta_{\text{BF}}(N_\text{t}, K) = \mathcal{O}\big( K^2 N_\text{t}^{\frac{7}{2}} + K^3 N_\text{t}^{\frac{5}{2}} + K^4 N_\text{t}^{\frac{1}{2}} \big)$ are the computational complexities of the optimal resource allocation scheme in Algorithm~\ref{Alg:OptimalResourceAllocation} and the optimal energy beamforming design, i.e., the solution of (\ref{Eqn:OptimalVectorsRef}), respectively. Here, $\mathcal{O}(\cdot)$ denotes the big-O notation\footnotemark.
\footnotetext{The computational complexities of a linear program and a convex semidefinite problem that involve $n$ variables and $m$ constraints are given by $\mathcal{O}\big(m^3+nm^2+n^2 \big)$ and $\mathcal{O}\big(\sqrt{n}(mn^3+m^2n^2+m^3) \big)$, respectively \cite{Polik2010}.}
Thus, we conclude that the computational complexity of the optimal energy signal and resource allocation policy is exponential in the number of deployed users $K$ and polynomial in the number of transmit antennas $N_\text{t}$.
Hence, determining the optimal WPCN design may not be computationally efficient for multi-user systems with $K\gg1$.
Therefore, in the following section, we propose two suboptimal low-complexity schemes to solve optimization problem (\ref{Eqn:OriginalProblem}).

\begin{algorithm}[!t]		
	\small				
	\linespread{1.45}\selectfont
	\SetAlgoNoLine%
	\SetKwFor{Foreach}{for each}{do}{end}		
	Initialize: Channel vectors $\boldsymbol{h}_1, \boldsymbol{h}_2, \cdots, \boldsymbol{h}_K$, required rates $R_k^\text{req}$, power $p_k^\text{req}$, and initial energy $q_k$ at user $k, k \in \{1,2,\cdots, K\}$.	\\	
	1. Find the channel covariance matrix $\boldsymbol{B} = (\boldsymbol{H} \boldsymbol{H}^H)^{\frac{1}{2}}$ \\
	2. Determine the optimal time allocation $\bar{{\tau}}^*$, ${\boldsymbol{\tau}^\text{d}}^*$, and harvested powers $\boldsymbol{\mu}^*_n, n\in\{1,2,\cdots, \bar{N}\},$ with Algorithm~1\\
	3. For the optimal harvested powers, determine the energy beamforming vectors $\boldsymbol{w}_n^*$ in Proposition~\ref{Theorem:MuProp4}\\
	\textbf{Output:} Energy beamforming vectors $\boldsymbol{w}_n^*, n\in \{1,2,\cdots, \bar{N}\}$, ratio $\bar{\tau}^*$, and time slot lengths $\boldsymbol{\tau}^*$
	\caption{\strut Optimal design of multi-user \gls*{wpcn}. }
	\label{Alg:OptimalSolution}
\end{algorithm}	 

\section{Low-Complexity Design of Multi-User WPCNs}
In this section, we propose two suboptimal low-complexity schemes for WPCN design.
To this end, we first consider asymptotic massive MISO WPCNs with $N_\text{t} \to \infty$ and determine the optimal transmit policy in the downlink in closed-form.
Next, based on this result, we propose an MRT-based scheme for WPCN design, which is optimal for massive MISO WPCNs with vanishing system loads and provides a suboptimal solution of (\ref{Eqn:OriginalProblem}) for general MISO WPCNs with finite system loads.
Finally, to improve the performance of this design for the general case, we derive a low-complexity suboptimal scheme that utilizes \gls*{sdr}.
\subsection{Massive MISO WPCNs}
\label{Section:MassiveMIMO}
In the following, we consider the optimal design of multi-user MISO WPCNs, where the channel vectors $\boldsymbol{h}_k, k\in\{1, 2, \cdots, K\},$ become orthogonal, i.e., $\boldsymbol{h}_i \boldsymbol{h}_j^H = 0, \forall i \neq j$, as the WPCN load approaches zero, i.e., $\frac{K}{N_\text{t}} \to 0$ \cite{Yang2015}.
We note that this property holds for, e.g., massive MISO WPCNs with Rayleigh fading channels as the number of antennas $N_\text{t}$ at the BS tends to infinity \cite{Yang2015, Interdonato2020}. 

To solve optimization problem (\ref{Eqn:OriginalProblem}), we first determine the optimal fraction $\bar{\tau}^*$ using a one-dimensional grid search as $\bar{\tau}^* = \argmin_{\bar{\tau} \in [0, 1]} P_\text{DL}^*(\bar{\tau})$, where
\vspace*{-7pt}	
\begin{equation}
	\vspace*{-7pt}
	P_\text{DL}^*(\bar{\tau}) = \min_{\mathcal{F}} \{P_{\text{DL}} \}
	\label{Eqn:ProblemRef1}
\end{equation}
with $\mathcal{F} = \{ \boldsymbol{\mathcal{W}}, \boldsymbol{\tau}^\text{d}: \sum_{n=1}^{\bar{N}} \tau^\text{d}_{n} \phi_k\big(|\boldsymbol{h}_k \boldsymbol{w}_n|^2\big) \geq \xi^\text{d}_k(\bar{\tau}), \forall k; \; \sum_{n=1}^{\bar{N}} \tau^\text{d}_{n} = \bar{\tau} \}$, $\boldsymbol{\mathcal{W}} = \{\boldsymbol{w}_1, \boldsymbol{w}_2, \cdots, \boldsymbol{w}_{\bar{N}} \}$, and $P_{\text{DL}} = \sum_{n=1}^{\bar{N}} \tau^\text{d}_{n} \, \| \boldsymbol{w}_n \|_2^2$.

Next, for a given $\bar{{\tau}}$, in the following proposition, we show that for massive MISO WPCNs, the optimal energy signal and time allocation policy as solution of (\ref{Eqn:RefProblem}) can be obtained in closed-form.
\begin{proposition}
	If channel vectors $\boldsymbol{h}_k, k\in\{1, 2, \cdots, K\},$ are orthogonal, i.e., $\boldsymbol{h}_i \boldsymbol{h}_j^H = 0, \forall i \neq j$, for a given $\bar{\tau}$, the $K$ energy beamforming vectors of the optimal energy signal in Proposition~\ref{Theorem:MuProp1} and the normalized lengths of the corresponding time slots are given by
	\begin{align}
		\boldsymbol{w}_n^* &= \sum_{i=n}^{K} A_{k_i} \frac{ \boldsymbol{h}_{k_i}^H }{ \|\boldsymbol{h}_{k_i}\|^2_2 }, \\
		{\tau^\text{\upshape d*}_{n}} &= \bar{{\tau}} \big( \bar{t}^{+}_{k_n}(\bar{\tau}) - \sum_{i=1}^{n-1} \bar{t}^{+}_{k_i}(\bar{\tau}) \big), n \in \{1,2, \cdots, K\},
		\label{Eqn:SolutionMassiveMimo}
	\end{align}
	\noindent respectively.
	Here, the elements of $\bar{\boldsymbol{t}}^{+}(\bar{\tau}) = [\bar{t}^{+}_{k_1}(\bar{\tau}), \bar{t}^{+}_{k_2}(\bar{\tau}), \cdots, \bar{t}^{+}_{k_K}(\bar{\tau})] \in \mathbb{R}_{+}^K$ are obtained by sorting the values $t_k^{+}(\bar{\tau}) = \max \big\{ 0, \frac{\xi^\text{\upshape d}_k(\bar{\tau})}{A_k^2} \big\}, k\in\{1,2,\cdots, K\}$, in ascending order.
	Furthermore, the $(K+1)^\text{th}$ energy beamforming vector is the all-zero vector, i.e., $\boldsymbol{w}^*_{K+1} = \boldsymbol{0}$, and its normalized length is given by $\tau^\text{d*}_{K+1} = \bar{\tau} - \sum_{n=1}^{K} \tau^\text{d*}_{n}$.
	Finally, the corresponding optimal transmit power in the downlink is given by $P^*_{\text{DL}}(\bar{{\tau}}) =  \sum_{k=1}^{K} \frac{ \xi^\text{d}_k(\bar{\tau}) }{ \|\boldsymbol{h}_{k}\|^2_2 }$.
	\label{Theorem:MassiveMIMO}
\end{proposition}
\begin{proof}
	Please refer to Appendix~\ref{Appendix:Prop5}.
\end{proof}

\begin{figure}[t]
	\centering
	\includegraphics[draft=false, width=0.35\textwidth]{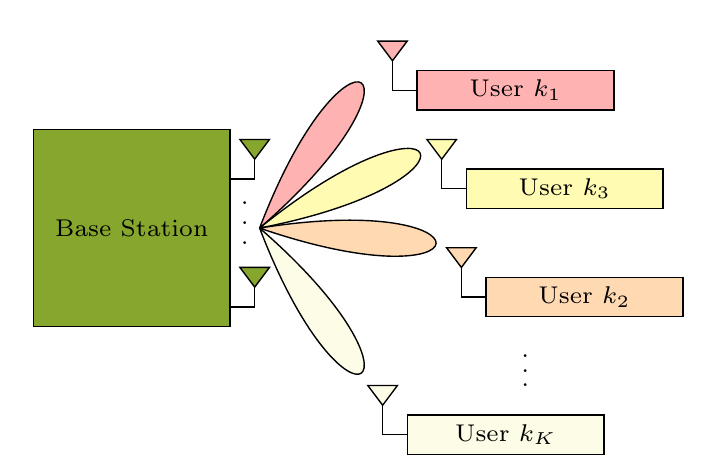}
	\includegraphics[draft=false, width=0.64\textwidth]{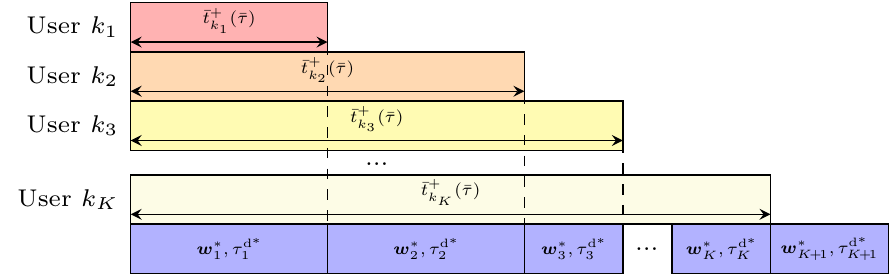}
	\caption{Optimal design of massive MISO WPCNs. The optimal energy beamforming vectors in the downlink are weighted sums of MRT beamformers, where the weights are chosen such that the corresponding EH circuits are driven into saturation.}
	\vspace*{-10pt}
	\label{Fig:OptimalMassiveMiso}
\end{figure}

Proposition~\ref{Theorem:MassiveMIMO} reveals that for multi-user massive MISO WPCNs, the optimal energy signal comprises a sequence of weighted sums of MRT beamforming vectors and can be obtained in closed-form.
Furthermore, similar to the single-user WPCN in Section~\ref{Section:SingleUser}, the magnitudes of the MRT beamforming vectors in $\boldsymbol{w}_n^*, n \in \{1,2,\cdots, \bar{N}\},$ are chosen to drive the corresponding EH circuits at the user devices into saturation.
Thus, in the first time slot, it is optimal to drive all EH circuits into saturation.
The corresponding normalized time slot duration depends on the power required by the least demanding user $k_1$, i.e., the user with the minimum normalized required power $\frac{\xi^\text{d}_k(\bar{\tau})}{A_k^2}, k\in\{1,2,\cdots, K\}$.
Furthermore, since the channel vectors $\boldsymbol{h}_k, k\in\{1, 2, \cdots, K\},$ are orthogonal, for the optimal resource allocation policy, in each subsequent time slot, it is optimal to switch-off the power transmission to the user whose power requirement is already satisfied, cf. Fig.~\ref{Fig:OptimalMassiveMiso}.
Finally, in the $(K+1)^\text{th}$ time slot, all constraints in (\ref{Eqn:RefProblem}) are satisfied and, hence, the corresponding optimal beamforming vector is the all-zero vector.
We note that since the time reserved for the last time slot can be utilized for uplink information transmission, for the optimal $\bar{\tau}^*$, we have $\tau^\text{d*}_{K+1} = 0$.

\begin{remark}
	The optimal downlink energy signal for massive MISO systems as solution of problem (\ref{Eqn:OriginalProblem}) may not be unique.
	For example, the order in which $\boldsymbol{w}^*_1, \boldsymbol{w}^*_2, \cdots, \boldsymbol{w}^*_{K+1}$ are transmitted can be chosen arbitrarily.
	Furthermore, the optimal normalized durations of the time slots may differ from those in Proposition~\ref{Theorem:MassiveMIMO} if, e.g., the time slot of user $k_1$ is split into several parts to serve this user in multiple non-consecutive time slots.
\end{remark}

In the following, we exploit Proposition~\ref{Theorem:MassiveMIMO} for the derivation of suboptimal WPCN designs.
In particular, we first propose a closed-form suboptimal MRT-based WPCN design that is optimal for massive MISO WPCNs with vanishing system loads and is a feasible point of optimization problem (\ref{Eqn:OriginalProblem}) for multi-user MISO WPCNs with finite system loads.
Next, exploiting \gls*{sdr}, we derive a low-complexity suboptimal scheme that is based on the MRT-based WPCN design and yields a scheme with reduced average transmit power at the BS required in the general case.

\subsection{Suboptimal MRT-based Scheme}
In this section, based on the optimal energy beamforming vectors for massive MISO WPCNs with vanishing system loads in Proposition~\ref{Theorem:MassiveMIMO}, we propose a suboptimal MRT-based WPCN design that is optimal for massive MISO WPCNs and solves (\ref{Eqn:OriginalProblem}) for general multi-user WPCNs.
We note that the scheme in Proposition~\ref{Theorem:MassiveMIMO} may not be a feasible solution of (\ref{Eqn:OriginalProblem}) for the general case since the energy beamforming vectors in Proposition~\ref{Theorem:MassiveMIMO} may not be able to drive the corresponding EH circuits into saturation if the channel vectors are not orthogonal.
Thus, in the following, we adapt the energy beamforming vectors in Proposition~\ref{Theorem:MassiveMIMO} to the case of general WPCNs with finite system loads.

To this end, for a given $\bar{{\tau}}$, we first determine the energy beamforming vectors ${\boldsymbol{w}}_n^*$ and the normalized lengths of the corresponding time slots $\tau^\text{d*}_{n}, n \in \{1,2,\cdots, K+1\}$, as in Proposition~\ref{Theorem:MassiveMIMO}.
Next, to drive the EH circuits of users $k \in \mathcal{K}(n) = \{k_n, k_{n+1}, \cdots, k_K\}$ into saturation, in time slot $n$, we adopt a weighted energy beamforming vector $\tilde{\boldsymbol{w}}_n^* = \omega_n {\boldsymbol{w}}_n^*$.
We propose to choose the weights $\omega_n, n\in\{1,2,\cdots,\bar{N}\},$ as follows:
\vspace*{-5pt}
\begin{equation}
	\vspace*{-5pt}
	\omega_n = \max_{k'\in \mathcal{K}(n)}  \frac{ A_{k'} }{| \boldsymbol{h}_{k'} \boldsymbol{w}^*_n |}.
\end{equation}
\noindent We note that since $\boldsymbol{w}^*_n$ in Proposition~\ref{Theorem:MassiveMIMO} is a weighted sum of random channel vectors $\boldsymbol{h}_{k'}$ of users $k' \in \mathcal{K}(n)$, we have $| \boldsymbol{h}_{k'} \boldsymbol{w}^*_n | > 0$ and $\omega_n < \infty, \forall n$, with probability $1$.
Furthermore, in the case of orthogonal channel vectors, we have $\omega_n = 1, \forall n\in\{1,2,\cdots, K\},$ and $\tilde{\boldsymbol{w}}_n^* = \boldsymbol{w}_n^* $.
The algorithm proposed to obtain the suboptimal MRT-based WPCN design is summarized in Algorithm~\ref{Alg:InitialPoint}.

In the proposed suboptimal MRT-based scheme, similar to Proposition~\ref{Theorem:MassiveMIMO}, energy beamforming vector $\tilde{\boldsymbol{w}}_n^*$ in time slot $n$, $n\in\{1,2,\cdots, \bar{N}\},$ is chosen to drive the EH circuits of the users in $\mathcal{K}(n)$ into saturation and, thus, the proposed scheme is a feasible solution of (\ref{Eqn:OriginalProblem}).
However, this scheme may not be efficient for general multi-user WPCNs with finite system loads since the EH circuits at the user devices may be driven too deep into saturation, i.e., $| \boldsymbol{h}_{k} \tilde{\boldsymbol{w}}_n^* | ^2 > A_k^2$ for some $n\in \{1,2,\cdots, K\}$ and $k\in \{1,2,\cdots, \bar{N}\},$ and thus, the constraints in  (\ref{Eqn:ProblemC1}) and (\ref{Eqn:ProblemC2}) may not be tight.
Therefore, to cope with this issue, in the next section, we propose a low-complexity suboptimal scheme that is initialized by the solution attained with the suboptimal MRT-based scheme and yields a scheme with reduced average transmit power in the energy transfer phase for general WPCNs.

\begin{algorithm}[!t]	
	\small				
	\SetAlgoNoLine%
	\SetKwFor{Foreach}{for each}{do}{end}		
	Initialize: Required rates $R_k^\text{req}$, powers $p_k^\text{req}$, initial energies $q_k, \forall k$, error tolerance $\epsilon_{\tau}$	\\	
		1. Set initial value $\bar{\tau} = 0$, $i = 1$.\\
	\Repeat{$\bar{\tau} \geq \bar{\tau}^\text{\upshape{max}}$}{		
		2. Determine vector $\boldsymbol{t}^{+}(\bar{\tau})$ as in Proposition~\ref{Theorem:MassiveMIMO}\\
		3. Find $\boldsymbol{w}_n^*$ and $\tau_n^\text{d*}$, $n\in\{1,2,\cdots, \bar{N}\},$ as in Proposition~\ref{Theorem:MassiveMIMO}\\
		4. Obtain $\tilde{\boldsymbol{w}}_n = \omega_n \boldsymbol{w}_n^*, n\in\{1,2,\cdots, \bar{N}\},$ with $\omega_n = \max_{k'\in \mathcal{K}(n)}  \frac{ A_{k'} }{| \boldsymbol{h}_{k'} \boldsymbol{w}_n |}$ \\
		5. Save $\tilde{\boldsymbol{w}}_n$ and $\tau_n^\text{d*}$, $n\in\{1,2,\cdots, \bar{N}\},$ in $\boldsymbol{\mathcal{W}}_i$ and $\boldsymbol{\tau}^\text{d}_i$, respectively\\
		6. Calculate $\eta_{i} = \sum_{n=1}^{\bar{N}} \tau^\text{d*}_{n} \, \| \tilde{\boldsymbol{w}}^*_n \|_2^2$\\
		7. Set $\bar{\tau} = \bar{\tau} + \epsilon_{\tau}$, $i = i+1$
	}
	8. Find index $i^*$ yielding the minimum value in $\boldsymbol{\eta} = [\eta_1, \eta_2, \cdots, \eta_{i-1}]$ \\
	\textbf{Output:} 
	$\tilde{\boldsymbol{w}}_n, \tau_n^\text{d*}, \forall n \in \{1,2,\cdots, \bar{N}\},$ from $\boldsymbol{\mathcal{W}}_{i^*}$ and $\boldsymbol{\tau}_{i^*}$, respectively, and $\bar{\tau}^* = \epsilon_{\tau}(i^*-1)$
	\caption{\strut Algorithm for Suboptimal MRT-based Scheme}
	\label{Alg:InitialPoint}
\end{algorithm}	
\subsection{Suboptimal SDR-based Scheme}
Since the MRT-based WPCN design obtained with Algorithm~\ref{Alg:InitialPoint} may be severely suboptimal if the system load $\frac{K}{N_\text{t}}$ does not tend to $0$, in this section, we improve the performance of the proposed suboptimal design for the general case.
In \cite{Shanin2021d}, we developed an iterative algorithm that converges to a stationary point of (\ref{Eqn:OriginalProblem}).
However, since the complexity of this algorithm may be unreasonably high, in this work, we propose a low-complexity solution of (\ref{Eqn:OriginalProblem}) that is based on \gls*{sdr} and the MRT-based scheme, which is optimal for massive MISO WPCNs with vanishing system loads.

Since the beamforming vectors $\tilde{\boldsymbol{w}}_n^*, n\in\{1,2,\cdots, \bar{N}\}$, obtained with Algorithm~\ref{Alg:InitialPoint} may drive the EH circuits at the user devices too deep into saturation, i.e., $|\boldsymbol{h}_k \tilde{\boldsymbol{w}}_n^*|^2 \gg A_k^2$ for some $k\in\{1,2,\cdots, K\}$ and $n \in \{1,2,\cdots, \bar{N}\}$, we propose to modify energy beamforming vectors $\tilde{\boldsymbol{w}}_n^*$ such that the modified vectors $\hat{\boldsymbol{w}}_n^*, n\in\{1,2,\cdots, \bar{N}\},$ provide the same harvested powers to the users, but avoid the waste of energy and thus, reduce the average transmit power at the BS.
To this end, in time slot $n$, we first determine the harvested powers $\tilde{\mu}_{n,k} = \phi_k(|\boldsymbol{h}_k \tilde{\boldsymbol{w}}_n^*|^2)$ that are achieved at user devices $k \in \{1,2,\cdots, K\}$ with the suboptimal MRT-based scheme.
Next, we determine new energy beamforming vectors $\hat{\boldsymbol{w}}_n^*, n\in\{1,2,\cdots, \bar{N}\},$ such that $\| \hat{\boldsymbol{w}}_n^*\|_2 \leq \| \tilde{\boldsymbol{w}}_n^*\|_2$ and $\phi_k(|\boldsymbol{h}_k \hat{\boldsymbol{w}}_n^*|^2) \geq \tilde{\mu}_{n,k}, \forall n, k$.
To this end, in time slot $n$, we solve the following optimization problem:
\begin{equation}
	\minimize_{\hat{\boldsymbol{W}}_n \in \mathcal{S}^{N_\text{t}}_{+}} \quad \text{\upshape{Tr}}\{\hat{\boldsymbol{W}}_n\} \quad
	\subjectto \quad \boldsymbol{h}_k \hat{\boldsymbol{W}}_n \boldsymbol{h}_k^H \geq \phi_k^{-1}(\tilde{\mu}_{n,k}), \; \forall k.
	\label{Eqn:VectorsSub2}
\end{equation} 
Since (\ref{Eqn:VectorsSub2}) is a convex semidefinite optimization problem, it can be solved with a numerical optimization tool, such as CVX \cite{Grant2015}.
Furthermore, similar to problem (\ref{Eqn:OptimalVectorsRef}), the optimal solution $\hat{\boldsymbol{W}}_n^*$ of problem (\ref{Eqn:VectorsSub2}) satisfies $\rank\{\hat{\boldsymbol{W}}_n^*\} \leq 1, \forall n \in\{1,2,\cdots, \bar{N}\}$, cf. Lemma~\ref{Theorem:Lemma}.
Hence, the beamforming vector in time slot $n, n \in \{1,2,\cdots, \bar{N}\}$, can be obtained as $\hat{\boldsymbol{w}}_n^* = \hat{\gamma}^*_n \hat{\boldsymbol{v}}^*_n$, where $\hat{\gamma}^*_n$ and $\hat{\boldsymbol{v}}^*_n$ are the dominant eigenvalue and the corresponding eigenvector of matrix $\hat{\boldsymbol{W}}^*$, respectively, see Proposition~\ref{Theorem:MuProp4}.

The resulting suboptimal SDR-based design is summarized in Algorithm~\ref{Alg:SubScheme2}.
We note that since the suboptimal MRT- and SDR-based designs provide identical harvested powers at the user devices, the solution obtained with Algorithm~\ref{Alg:SubScheme2} is also a feasible point of problem (\ref{Eqn:OriginalProblem}).
Furthermore, since problem (\ref{Eqn:VectorsSub2}) is similar to (\ref{Eqn:OptimalVectorsRef}), the computational complexity of the suboptimal SDR-based scheme is given by $\Theta_{\text{Sch.2}}(N_\text{t}, K) = \Theta_{\text{BF}}(N_\text{t}, K) = \mathcal{O}\big( K^2 N_\text{t}^{\frac{7}{2}} + K^3 N_\text{t}^{\frac{5}{2}} + K^4 N_\text{t}^{\frac{1}{2}} \big).$
Thus, the computational complexity of the proposed algorithm is polynomial in the number of users $K$ and the number of transmit antennas $N_\text{t}$ at the BS.

\begin{algorithm}[!t]	
	\small				
	\SetAlgoNoLine%
	\SetKwFor{Foreach}{for each}{do}{end}		
	Initialize: Channel vectors $\boldsymbol{h}_1, \boldsymbol{h}_2, \cdots, \boldsymbol{h}_K$, required rates $R_k^\text{req}$, powers $p_k^\text{req}$, and initial energy $q_k$ at user $k, k \in \{1,2,\cdots, K\}$.	\\	
	1. Find $\tilde{\boldsymbol{w}}_n, \tau_n^{d*}, n\in\{1,2,\cdots, \bar{N}\}$, and $\bar{{\tau}}^*$ with Algorithm~\ref{Alg:InitialPoint} \\
	2. Determine the optimal matrices $\hat{\boldsymbol{W}}^*_n, n\in\{1,2,\cdots, \bar{N}\},$ as solutions of (\ref{Eqn:VectorsSub2})\\
	3. Obtain $\hat{\boldsymbol{w}}_{n}^* = \hat{\gamma}^*_n \hat{\boldsymbol{v}}_{n}^*, \forall n$\\
	\textbf{Output:} 
	$ \bar{\tau}^*$, $\tau_n^{d*}$, $\hat{\boldsymbol{w}}_{n}^*, \forall n$
	\caption{\strut Algorithm for Suboptimal SDR-based Scheme. }
	\label{Alg:SubScheme2}
\end{algorithm}

\section{Numerical Results}
\label{Section:SimulationResults}
In this section, we evaluate the performance of the proposed schemes for WPCN design via numerical simulations.
First, we discuss the adopted system setup. 
Next, we study the complexity of the proposed optimal and suboptimal schemes.
Finally, we analyze the performances of the proposed schemes and compare them with baseline schemes.
\subsection{Simulation Setup}
\begin{table*}[!t]
	\centering
	\caption{Simulation parameters}
	\begin{tabular}{|c|l|c|l|}
		\hline \\[-1.5em]
		\multicolumn{1}{|c|}{EH model in (\ref{Eqn:EhModel})} & \multicolumn{3}{|c|}{$\varphi_k(|z|^2) = \lambda \Big[\mu^{-1} W_0\Big( \mu \exp(\mu) I_0 \Big( \nu \sqrt{2 |z|^2} \Big) - 1 \Big)^2\Big], \; \forall k$}  \\
		\hline
		\makecell{Parameters of \\ EH model} & \makecell[l]{$\mu=0.03$, $\nu = 2.4\cdot10^3$, \\ $\lambda = 10^{-10}$, $A_k^2 = \SI{0.4}{\milli\watt}, \forall  k$}  & 
		\makecell{Grid sizes for Algorithm 1 } & \makecell[l]{$L_{\mu} = 10$ \\$L_\tau = 100$ }
		\\
		\hline
		\makecell{Distance} & \makecell[l]{In Fig.~\ref{Fig:Results_Complexity}: $d_k = \SI{3}{\meter}, \; \forall k$ \\ In Fig.~\ref{Fig:Results_PowersRates}: $d_1 = \SI{3}{\meter}, d_2 = \SI{5}{\meter}, d_3 = \SI{7}{\meter}$ \\ In Fig.~\ref{Fig:Results_K_Nt}: $d_k \in [\SI{3}{\meter}, \SI{10}{\meter}], \; \forall k$}  & 
		\makecell{Step size for Algorithm 3 \\ Initial energy} & \makecell[l]{$\epsilon_{\tau} = 10^{-2}$ \\ $q_k = \SI{0}{\joule},\; \forall k$} 
		\\
		\hline
	\end{tabular}
	\label{Table:SimulationSetup}
\end{table*}
We set the path loss between the BS and each user device as $\big(\frac{c_l}{4 \pi d_k f_c}\big)^2$, where $c_l$ is the speed of light, $f_c = \SI{868}{\mega\hertz}$ is the carrier frequency, and $d_k$ is the distance between the BS and user $k \in \{1,2,\cdots, K\}$.
Next, for information transmission in the uplink, we assume a noise variance of $\sigma_k^2 = \SI{-120}{\dBm}$.
We assume that the channel gains in $\boldsymbol{h}_k, k\in\{1,2,\cdots,K\},$ follow independent Rayleigh distributions.
Additionally, we assume that all user devices are equipped with identical EH circuits and for the EH model $\phi_k(\cdot)$ in (\ref{Eqn:EhModel}), we adopt the non-linear function $\varphi_k(|z|^2) = \lambda \Big[\mu^{-1} W_0\Big( \mu \exp(\mu) I_0 \Big( \nu \sqrt{2 |z|^2} \Big) - 1 \Big)^2\Big]$ derived in \cite{Morsi2019} for a half-wave rectifier with a single diode \cite{Tietze2012}.
Here, $\lambda, \mu$, and $\nu$ are parameters that depend on the circuit elements but not on the received signal, whereas $W_0(\cdot)$ and $I_0(\cdot)$ are the principle branch of the Lambert-W function and the modified Bessel function of the first kind and order zero, respectively.
The simulation parameters are summarized in Table~\ref{Table:SimulationSetup}.
All numerical results are averaged over 1000 random channel realizations.
\subsection{Complexity Analysis}
\begin{figure}[!t]
	\centering
	\includegraphics[draft=false, width=0.45\textwidth]{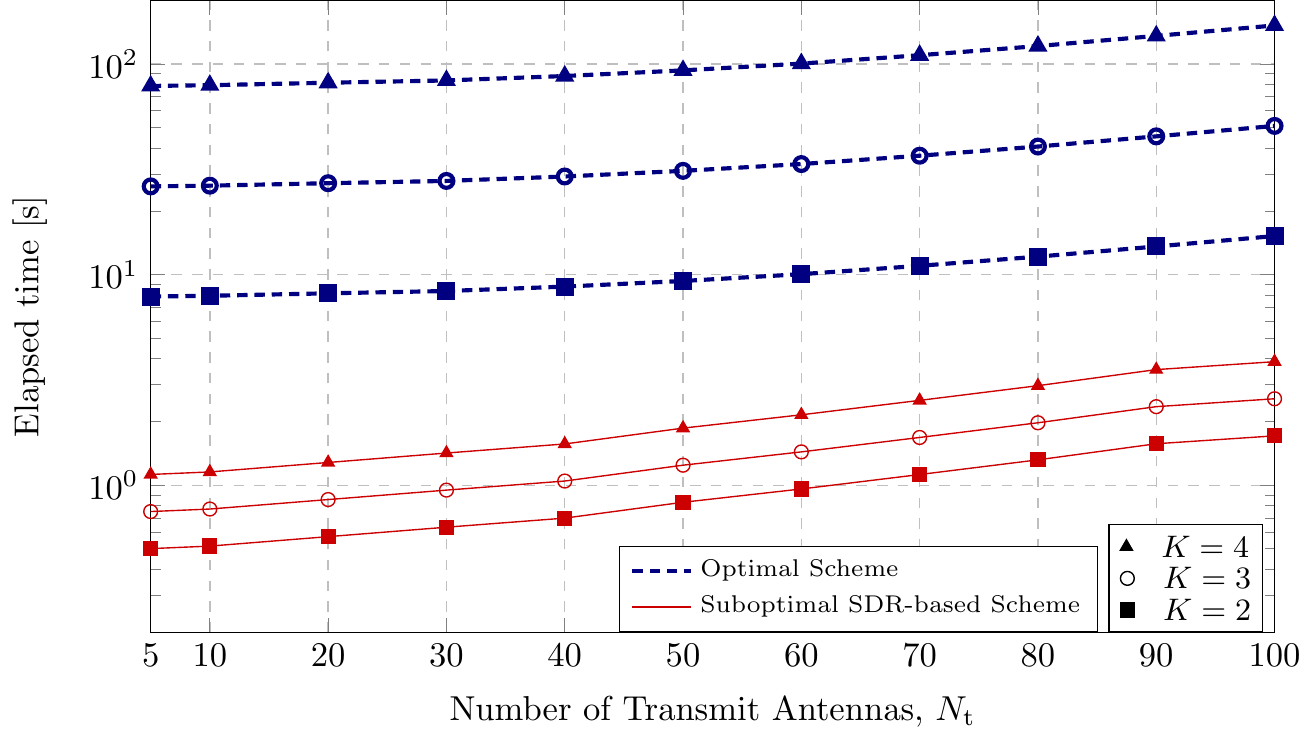}
	\vspace*{-10pt}
	\caption{Computational complexities of the optimal scheme and suboptimal SDR-based design for different values of $N_\text{t}$ and $K$.}
	\label{Fig:Results_Complexity}
	\vspace*{-10pt}
\end{figure}

First, we study the complexities of the proposed algorithms.
In Fig.~\ref{Fig:Results_Complexity}, we compare the times required to determine the energy signals with the optimal scheme and the suboptimal SDR-based scheme.
For the simulations, we adopt $p^\text{req}_k = \SI{0}{\watt}$ and the required per-user data rate is selected randomly from the interval $R_k \in [1,5] \SI{}{\frac{\bit}{channel\,use}}, \forall k$. 
In Fig.~\ref{Fig:Results_Complexity}, we observe that the computational complexities of both considered schemes grow with the number of users $K$.
Furthermore, we note that the time required for determining the optimal solution of (\ref{Eqn:OriginalProblem}) is strongly dominated by the exponential complexity in $K$.
In contrast, the suboptimal WPCN design obtained with Algorithm~\ref{Alg:SubScheme2} entails a computational complexity that is polynomial in both $K$ and $N_\text{t}$.
Thus, we observe that for any value of $K$, the suboptimal SDR-based scheme entails a significantly lower complexity than the optimal WPCN design.
Moreover, although the computational complexity of the optimal and suboptimal SDR-based WPCN designs may be high for large values of $N_\text{t}$, we note that in this case, both schemes may not provide substantial gains over the suboptimal MRT-based scheme which is optimal for massive MISO WPCNs with vanishing system loads and whose computational complexity is very small.

\subsection{Performance Analysis}
\begin{figure}[!t]
\centering
	\subfigure[Comparison for different per-user required powers $p^\text{req}$]{
		\includegraphics[draft=false, width=0.47\textwidth]{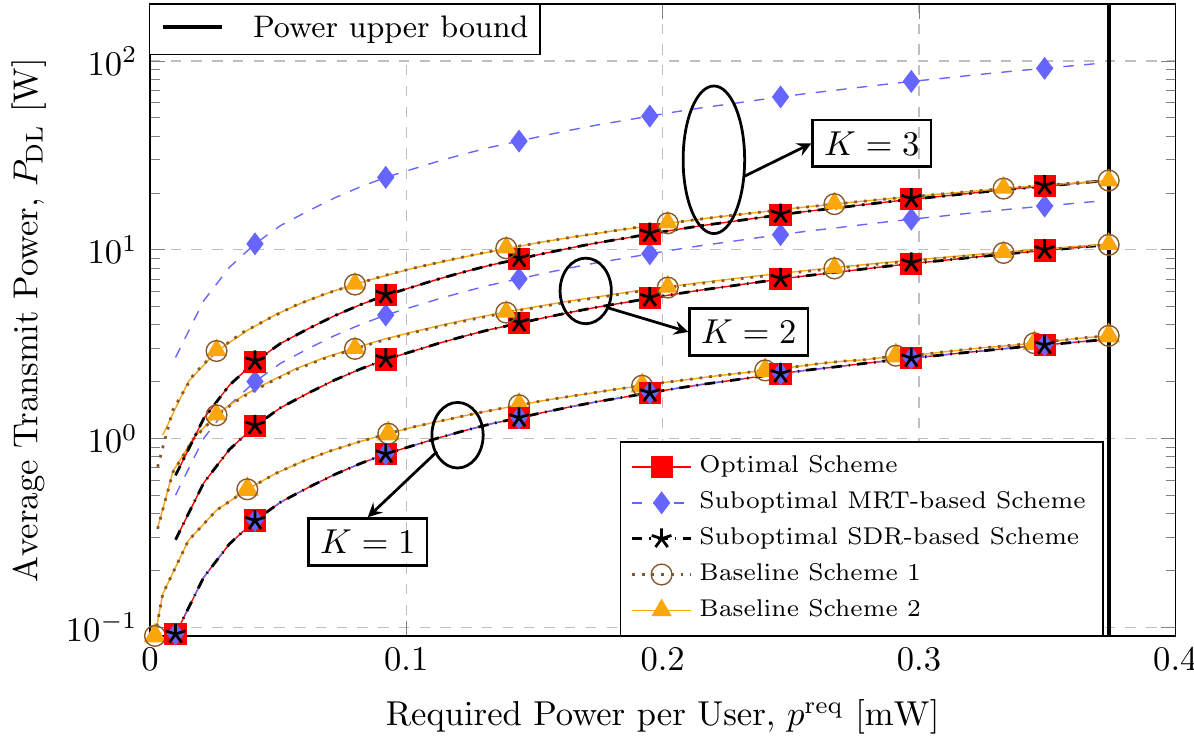} \label{Fig:Results_Powers}}
	\subfigure[Comparison for different per-user required rates $R^\text{req}$]{
		\includegraphics[draft=false, width=0.47\textwidth]{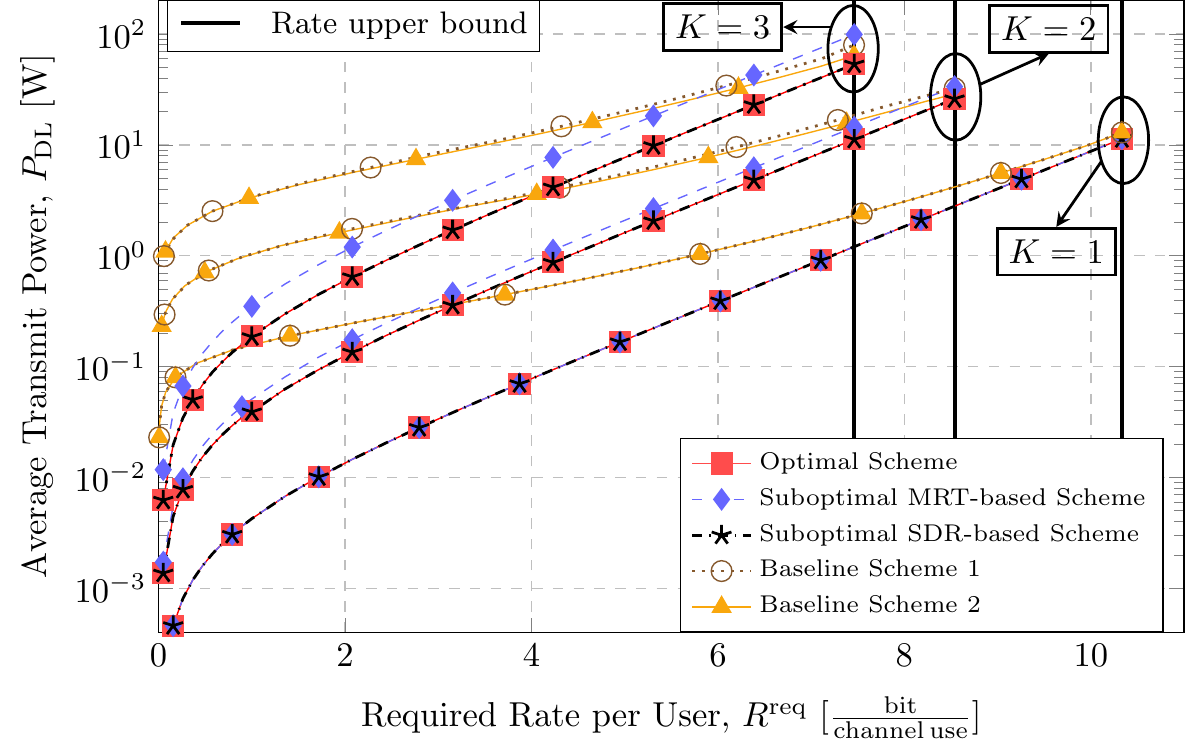}
		\label{Fig:Results_Rates}}
	\vspace*{-5pt}
		\caption{Average transmit powers $P_\text{DL}$ for different values of per-user required powers $p^\text{req}$ and rates $R^\text{req}$.}
\label{Fig:Results_PowersRates}
	\vspace*{-7pt}
\end{figure}

In the following, we analyze the performances of the proposed energy signal designs for WPCNs.
To this end, in Figs.~\ref{Fig:Results_Powers} and \ref{Fig:Results_Rates}, we compare the average transmit powers $P_\text{DL}$ obtained for different numbers of users $K$ as functions of the required per-user power, $p^\text{req} = p_k^\text{req}, \; \forall k,$ and the required per-user rate, $R^\text{req} = R_k^\text{req}, \; \forall k$, respectively.
We adopt $N_\text{t} = 5$ antennas at the BS and consider WPCNs with $K=1$, $K=2$, and $K=3$ user devices.
We set the distances between the BS and user devices to $d_1 = \SI{3}{\meter}$, $d_2 = \SI{5}{\meter}$, and $d_3 = \SI{7}{\meter}$, respectively.
Furthermore, for the results in Figs.~\ref{Fig:Results_Powers} and \ref{Fig:Results_Rates}, we adopt $R_k^\text{req} = \SI{0}{\frac{\bit}{channel\,use}}$ and $p_k^\text{req} = \SI{0}{\watt}$, $\forall k$, respectively.

As Baseline Scheme 1 and Baseline Scheme 2, we adopt the WPCN design obtained by solving optimization problem (\ref{Eqn:OriginalProblem}) with a linear and sigmoidal EH models at the user devices, utilizing the algorithms reported in \cite{Liu2014} and \cite{Boshkovska2018}, respectively.
For the linear EH model in Baseline Scheme 1, we adopt energy conversion efficiency $\eta_k = \frac{\phi_k(A_k^2)}{A_k^2}, \forall k$, whereas for the sigmoidal model in Baseline Scheme 2, we determine the model parameters to accurately match the adopted function $\phi_k(\cdot), \forall k$, as proposed in \cite{Boshkovska2015}.
We note that, since both baseline EH models are designed to characterize the average harvested power at the user devices, the related signal designs provide covariance matrices $\tilde{\boldsymbol{X}} = \sum \frac{\tau_n}{\bar{\tau}} \boldsymbol{x}_n \boldsymbol{x}_n^H$ and assume Gaussian energy signal vectors $\boldsymbol{x}_n$.
Therefore, as in \cite{Liu2014}, for the baseline schemes, we could adopt $\tilde{N} = \rank \{\tilde{\boldsymbol{X}}\}$ transmit symbols vectors that are obtained from the $\tilde{N}$ dominant eigenvectors of $\tilde{\boldsymbol{X}}$.
However, in our extensive simulations, we always observed $\rank \{\tilde{\boldsymbol{X}}\} = 1$ and, thus, for both baseline schemes, we adopt a single transmit vector in the downlink $\boldsymbol{x}_1 = \tilde{\lambda}_1 \tilde{\boldsymbol{v}}_1$, where $\tilde{\lambda}_1$ and $\tilde{\boldsymbol{v}}_1$ are the non-zero eigenvalue and the corresponding eigenvector of $\tilde{\boldsymbol{X}}$, respectively.
Furthermore, since the baseline schemes may not be able to provide the required harvested powers and uplink rates if EH model in Table~\ref{Table:SimulationSetup} is adopted, for the baseline schemes, we plot in Figs.~\ref{Fig:Results_Powers} and \ref{Fig:Results_Rates} the average transmit powers for the achieved harvested powers and user rates, respectively.

First, we observe in Figs.~\ref{Fig:Results_Powers} and \ref{Fig:Results_Rates} that for any number of users $K$, the average transmit power $P_\text{DL}$ in the downlink increases with the required power $p^\text{req}$ and the required rate $R^\text{req}$.
Next, as expected, we observe that, for single-user WPCNs with $K=1$, all proposed schemes show identical performance since the suboptimal MRT-based scheme coincides with the optimal solution in this case.
We note that due to the saturation of the EH circuits, the achievable harvested power in Fig.~\ref{Fig:Results_Powers} and the uplink information rate in Fig.~\ref{Fig:Results_Rates} are upper bounded.
Furthermore, since ZF is adopted at the BS to mitigate inter-user interference, the upper bounds in Fig.~\ref{Fig:Results_Rates} are lower for larger numbers of users $K$.
We observe that all proposed schemes require a lower transmit power than the baseline schemes for any $K$, $p^\text{req}$, and $R^\text{req}$, respectively.
This is due to the more accurate modelling of the EH circuits which enables the optimization of the instantaneous powers harvested at the user devices, and thus, the optimal design of the downlink energy signal at the BS.
Furthermore, for the maximum achievable harvested powers and uplink information rates, the performance gap between the optimal and the baseline schemes is low since, in this regime, the optimal energy signal and resource allocation policy are determined by the weakest user and thus, a lower number of energy beamforming vectors may be optimal for downlink transmission.
Finally, although the suboptimal SDR-based scheme typically requires lower computational complexity than the optimal scheme, cf. Fig.~\ref{Fig:Results_Complexity}, we observe that the performances achieved by both schemes are practically identical for all considered $p^\text{req}$ and~$R_k^\text{req}$.

\begin{figure}[!t]
	\centering
	\subfigure[Comparison for different numbers of users $K$]{
		\includegraphics[draft=false, width=0.47\textwidth]{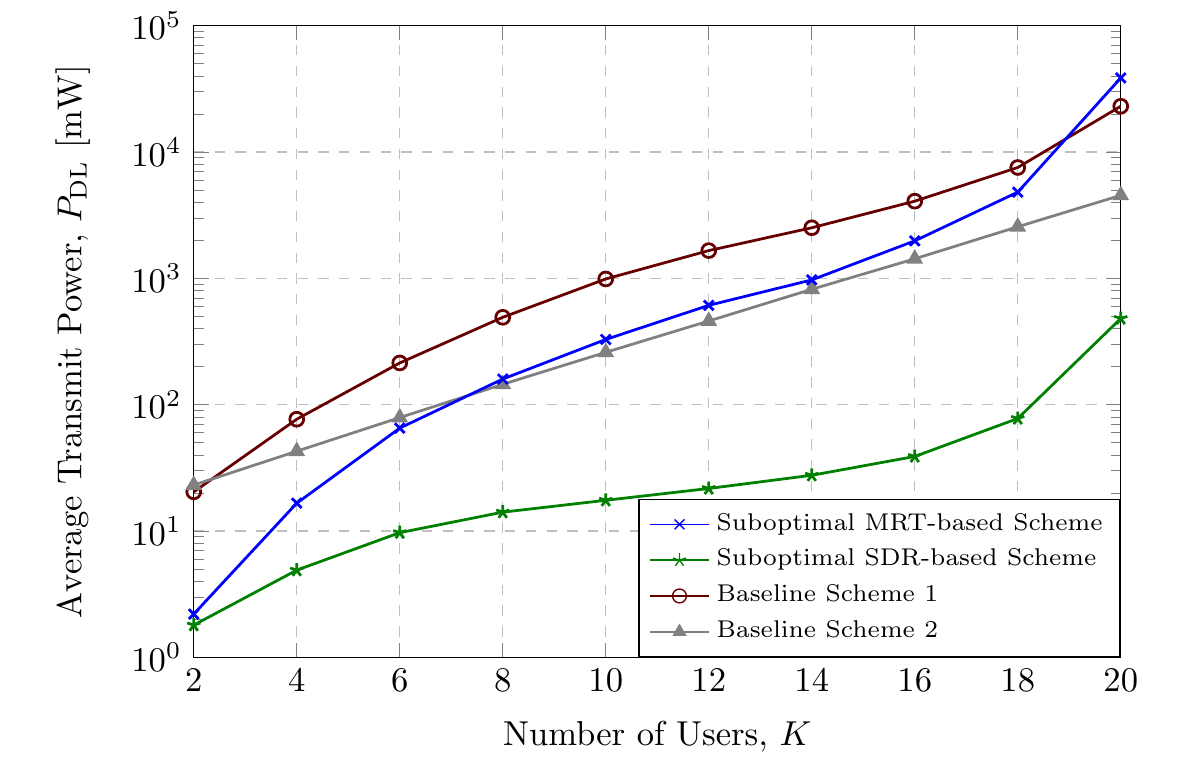} \label{Fig:Results_K}}
	\subfigure[Comparison for different numbers of antennas $N_\text{t}$]{
		\includegraphics[draft=false, width=0.47\textwidth]{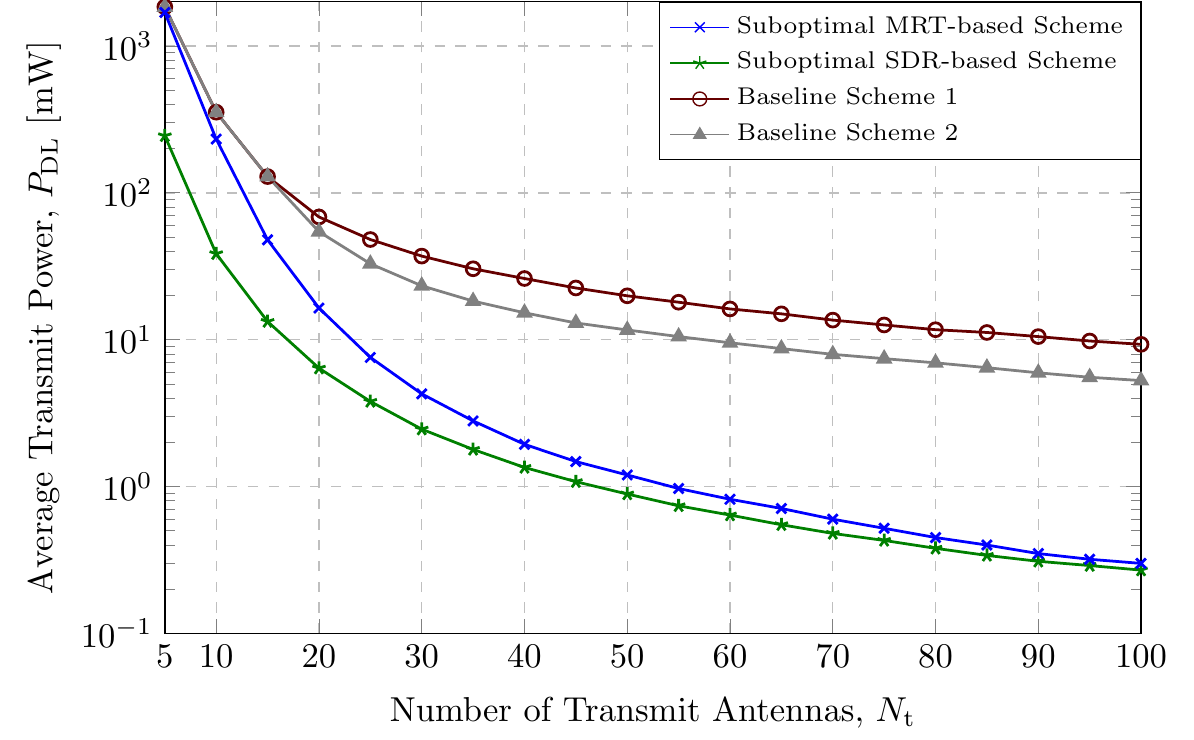}
		\label{Fig:Results_Nt}}
	\caption{Average transmit powers $P_\text{DL}$ for different numbers of BS antennas $N_\text{t}$ and users $K$.}
	\label{Fig:Results_K_Nt}
		\vspace*{-10pt}
\end{figure}
In Figs.~\ref{Fig:Results_K} and \ref{Fig:Results_Nt}, we depict the average transmit powers as functions of the numbers of users $K$ and BS antennas $N_\text{t}$, respectively.
We adopt $p_k^\text{req} = \SI{0}{\watt}$, $R_k^\text{req} = \SI{3}{\frac{\bit}{channel\,use}}, \forall k$, and $d_k$ is taken randomly from the interval $d_k \in [\SI{3}{\meter}, \SI{10}{\meter}], \forall k$.
For the results in Figs.~\ref{Fig:Results_K} and \ref{Fig:Results_Nt}, we adopt $N_\text{t} = 20$ antennas at the BS and $K = 5$ users, respectively.
Since the suboptimal SDR-based scheme has significantly lower computational complexity compared to the optimal scheme for large $K$, cf. Fig.~\ref{Fig:Results_Complexity}, and both schemes achieve identical performance, cf. Fig.~\ref{Fig:Results_PowersRates}, for the results in Fig.~\ref{Fig:Results_K_Nt}, we consider the suboptimal MRT- and SDR-based schemes only.

First, in Figs.~\ref{Fig:Results_K} and \ref{Fig:Results_Nt}, we observe that the required transmit power is smaller for lower numbers of users $K$ and larger numbers of BS antennas $N_\text{t}$, respectively.
This is expected since, for a given $K$, a higher number of transmit antennas, i.e., a smaller ratio $\frac{K}{N_\text{t}}$, leads to a larger beamforming gain and channel hardening, which yields better WPCN performance.
Next, we note that similar to the results in Fig.~\ref{Fig:Results_PowersRates}, the suboptimal SDR-based scheme achieves a significantly better performance compared to the baseline schemes due to the more accurate EH modelling which enables the optimization of the transmit energy signal waveform.
However, we observe that the suboptimal MRT-based scheme shows a poor performance and is even not able to outperform the baseline schemes if the number of deployed users is high.
Interestingly, the performance gap between the proposed suboptimal schemes decreases with increasing system load and for $N_\text{t} \gg K$, both proposed suboptimal schemes achieve nearly the same performance.
In fact, the MRT-based scheme becomes optimal for massive MIMO WPCNs when $N_\text{t} \to \infty$ and orthogonal channel vectors, i.e., $\boldsymbol{h}_i \boldsymbol{h}_j^H \to 0, \forall i \neq j$, and, thus, is more efficient for smaller system loads, i.e., ratios $\frac{K}{N_\text{t}}$.

\section{Conclusions}
In this work, we considered multi-user MISO WPCNs, where, in the downlink, the \gls*{bs} sent an energy signal comprising multiple energy transmit signal vectors to the users, which, in turn, harvested the received power and utilized it for information transmission in the uplink. 
To characterize the instantaneous power harvested at the user devices, we adopted a general non-linear EH model.
Then, we formulated a non-convex optimization problem, where we jointly designed the normalized durations of the downlink and uplink subframes and the energy signal waveform.
We showed that, for a single-user WPCN, the optimal signal in the downlink comprises a single energy signal vector which is collinear with the MRT beamformer and drives the EH circuit into saturation.
Then, to obtain the optimal solution for general multi-user WPCNs, we derived an optimal algorithm, whose computational complexity was exponential and polynomial in the number of users and BS antennas, respectively.
Moreover, we showed that the optimal energy signal for massive MISO WPCNs with vanishing system loads employs a sequence of weighted sums of the MRT beamforming vectors of the users.
Then, based on this solution, we proposed a closed-form MRT-based WPCN design, which is optimal for the massive MISO regime and suboptimal for MISO WPCNs with finite system loads.
Moreover, based on this scheme, we also derived a suboptimal SDR-based design that improved the performance of the MRT-based scheme for the general case.
Our simulation results revealed that the proposed suboptimal SDR-based scheme entails a lower complexity than the optimal design and both schemes achieve nearly identical performances and significantly outperform two baseline schemes based on linear and sigmoidal EH models, respectively.
Moreover, we observed that the performance gap between the proposed suboptimal MRT- and SDR-based schemes is small for lightly loaded WPCNs and becomes negligible when the number of BS antennas tends to infinity. 

\appendices
	\renewcommand{\thesection}{\Alph{section}}
	\renewcommand{\thesubsection}{\thesection.\arabic{subsection}}
	\renewcommand{\thesectiondis}[2]{\Alph{section}:}
	\renewcommand{\thesubsectiondis}{\thesection.\arabic{subsection}:}	
	\section{Proof of Proposition \ref{Theorem:SingleUser}}
	\label{Appendix:PropSU}
	First, we note that for any given transmit power, the received power at the user device is maximized if \gls*{mrt} beamforming vector $\boldsymbol{w}^*$ is adopted at the BS, i.e., $\forall \boldsymbol{x}_n \in \mathbb{C}^{N_\text{t}}$, $|\boldsymbol{h} \boldsymbol{x}_n| \leq |\boldsymbol{h} \boldsymbol{w}^* \alpha_n|$, where $\alpha_n = \| \boldsymbol{x}_n \|_2$.
Furthermore, since a scalar phase rotation of $\boldsymbol{x}_n$ does not influence the harvested power at the user device, the optimal energy signal vectors are given by $\boldsymbol{x}_n^* = \boldsymbol{w}^* s_n^*$, where $s_n = \alpha_n^* \exp(j \theta_{n})$ are scalar symbols with arbitrary phases $\theta_{n} \in [0, 2\pi)$.
Here, $\alpha_n^*$ are the optimal symbol magnitudes that solve the following optimization problem:
	\vspace*{-5pt}
\begin{equation}
		\minimize_{\substack{\boldsymbol{\tau}^\text{d}\succeq \boldsymbol{0}, \bar{\tau} \in[0,1], \boldsymbol{\alpha}, N}  } \, \sum_{n=1}^{N} \tau^\text{d}_n \alpha_n^2 \quad
		\subjectto \; \sum_{n=1}^{N} \tau^\text{d}_{n} \phi( \alpha_n^2 \| \boldsymbol{h} \|^2_2) \geq \xi(\bar{\tau}), \quad
		\sum_{n=1}^{N} \tau^\text{d}_{n} = \bar{\tau},
		\vspace*{-5pt}
	\label{Eqn:SUProblem_Ref1}
\end{equation}
\noindent\hspace*{-0pt}with $\boldsymbol{\alpha} = [\alpha_1, \alpha_2, \cdots, \alpha_N]$.

Next, we note that function $\phi(|z|^2)$ in (\ref{Eqn:SUProblem_Ref1}) is convex for $|z| \leq A$ and upper-bounded, i.e., $\phi(|z|^2) = \phi(A^2)$, $\forall |z|\geq A$.
Therefore, for any given value $\bar{\tau}$, the optimal solution\footnotemark\hspace*{0pt} to problem (\ref{Eqn:SUProblem_Ref1}) includes at most $N=2$ energy signal vectors, whose magnitudes are $\alpha_1^* = \alpha_s^*$ and $\alpha_2 = 0$, respectively, with corresponding time lengths $\tau^\text{d}_1 = \frac{\xi(\bar{\tau})}{\phi(A^2)}$ and $\tau^\text{d}_2 = \bar{{\tau}} - \tau^\text{d}_1$ \cite{Morsi2019, Shanin2021a}.
\footnotetext{The corresponding proof is similar to that of \cite[Corollary~3]{Shanin2021a} and is omitted here due to space constraints.}
Thus, problem (\ref{Eqn:SUProblem_Ref1}) can be further simplified as follows:
\vspace*{-5pt}
\begin{equation}
		\minimize_{\bar{\tau}\in[0,1], \tau^\text{d}_1 \geq 0  } \; \tau^\text{d}_1 \quad
		\subjectto \; \tau^\text{d}_{1} \phi( A^2 ) \geq \xi(\bar{\tau}), \; \tau^\text{d}_1 \leq  \bar{\tau}.
		\vspace*{-5pt}
	\label{Eqn:SUProblem_Ref2}
\end{equation}
Since a larger time slot duration $\bar{{\tau}}$ leads to a higher power requirement $(1-\bar{{\tau}})p^\text{u}_\text{min}$ for the uplink transmission phase, function $\xi(\bar{\tau})$ in (\ref{Eqn:SUProblem_Ref2}) is monotonic increasing.
Therefore, if problem (\ref{Eqn:SUProblem_Ref2}) is feasible, the optimal duration of the zero-valued downlink energy signal vector with $\alpha_2 = 0$ is $\tau^\text{d}_2 = 0$.
Thus, the optimal time allocation is achieved with $\tau^\text{d}_1 = \bar{\tau}$ and hence, one energy signal vector $\boldsymbol{x}^*$ is utilized at the BS, i.e., $N^* = 1$.
This concludes the proof.
	\section{Proof of Proposition \ref{Theorem:MuProp1}}
	\label{Appendix:Prop1}
	Let us reformulate optimization problem (\ref{Eqn:RefProblem}) equivalently as follows:
\begin{equation}
		\minimize_{\boldsymbol{\tau}^\text{d}, \bar{\tau}, \boldsymbol{X}, \boldsymbol{M}, N } \; \sum_{n=1}^{N} \tau^\text{d}_{n} \| \boldsymbol{x}_n \|_2^2 \;\;
		\subjectto \;  \sum_{n=1}^{N} \tau^\text{d}_{n} = \bar{\tau}, \; \sum_{n=1}^{N} \tau^\text{d}_{n} \boldsymbol{\mu}_n  \succeq \boldsymbol{\xi}^\text{d}(\bar{\tau}), \; \phi_k\big(|\boldsymbol{h}_k \boldsymbol{x}_n|^2\big) \geq \mu_{n,k}, \forall k,
\label{Eqn:ProblemProofProp1}
\end{equation}
\noindent\hspace*{-0pt}where $\boldsymbol{M} = [\boldsymbol{\mu}_1, \; \boldsymbol{\mu}_2, \; \cdots, \; \boldsymbol{\mu}_N]$.
Since the harvested powers at the user devices are independent of a scalar phase rotation of $\boldsymbol{x}_n, n\in\{1,2,\cdots, N\}$, we decompose the transmit energy signal vectors as in (\ref{Eqn:TransmitSymbolDecomposition}).
Next, we note that for any given vector of harvested powers in time slot $n$, $\boldsymbol{\mu}_n = [\mu_{n,1}, \mu_{n,2}, \cdots, \mu_{n,K}]^\top$, the optimal energy beamforming vector $\boldsymbol{w}_n^*$ that minimizes the instantaneous transmit power $\| \boldsymbol{w}_n \|^2_2$ and yields $\boldsymbol{\mu}_n$ can be obtained as in (\ref{Eqn:OptimalTransmitVectors}) and the corresponding minimum transmit power is given by $\psi(\boldsymbol{\mu}_n)$.
Finally, since for the optimal $\boldsymbol{x}_n^*$ and $\boldsymbol{\mu}_n^*$, we have $\|\boldsymbol{x}_n^*\|_2^2 = \psi(\boldsymbol{\mu}_n^*)$, problem (\ref{Eqn:ProblemProofProp1}) can be equivalently reformulated as (\ref{Eqn:OptimalResourceAlloc}). 
This concludes the proof.

	\section{Proof of Proposition \ref{Theorem:MuProp2}}
	\label{Appendix:Prop2}
	Since $\phi_k(\cdot)$ is monotonic non-decreasing and is given by (\ref{Eqn:EhModel}), optimization problem (\ref{Eqn:FunctionPsi}) can be equivalently reformulated as follows:
\vspace*{-10pt}
\begin{equation}
	\vspace*{-5pt}
	\psi(\boldsymbol{\mu}) = \min_{\boldsymbol{X} \in \Omega_{\mathcal{S}}(\boldsymbol{\mu})} \Tr{\boldsymbol{X}}
	\label{Eqn:FunctionPsiProof2}
\end{equation} 
\noindent with $\Omega_{\mathcal{S}}(\boldsymbol{\mu}) = \{\boldsymbol{X}: \boldsymbol{h}_k \boldsymbol{X} \boldsymbol{h}_k^H \geq \phi_k^{-1}(\mu_k), k\in\{1,2,\cdots,K\}, \forall k, \, \boldsymbol{X}\in\mathcal{S}^{N_\text{t}}_{+}, \, \rank \{\boldsymbol{X} \} \leq 1\}$, where $\boldsymbol{X} = \boldsymbol{x}\boldsymbol{x}^H$.
Furthermore, the matrix $\boldsymbol{X}^*$ that solves the following problem:
\vspace*{-5pt}
\begin{equation}
	\vspace*{-5pt}
		\minimize_{\boldsymbol{X}\in\mathcal{S}^{N_\text{t}}_{+} } \quad \Tr{\boldsymbol{X}} \quad
		\subjectto \; \boldsymbol{h}_k \boldsymbol{X} \boldsymbol{h}_k^H \geq  \phi_k^{-1}(\mu_k), \forall k,
	\label{Eqn:FunctionPsiProof2_1}
\end{equation} 
\noindent satisfies $\rank \{\boldsymbol{X}^*\} \leq 1$, cf. Lemma~\ref{Theorem:Lemma} in Section~\ref{Section:OptimalVectors}.
Hence, the value of function $\psi(\boldsymbol{\mu})$ can be equivalently obtained as $\psi(\boldsymbol{\mu}) = \Tr{\boldsymbol{X}^*}$, where $\boldsymbol{X}^*$ is the solution of problem (\ref{Eqn:FunctionPsiProof2_1}).

Let us consider optimization problem (\ref{Eqn:FunctionPsiProof2_1}).
We note that (\ref{Eqn:FunctionPsiProof2_1}) is a convex problem and, hence, its solution entails a computational complexity that is polynomial in $N_\text{t}$ \cite{Polik2010}.
Since we have $N_\text{t} \geq K$, to reduce the complexity of determining $\psi(\boldsymbol{\mu})$, in the following, we consider the dual problem associated with optimization problem (\ref{Eqn:FunctionPsiProof2_1}) and given by
\vspace*{-4pt}
\begin{equation}
		\vspace*{-4pt}
		\maximize_{\boldsymbol{\lambda} \succeq \boldsymbol{0}} \quad  \boldsymbol{\rho}^\top \boldsymbol{\lambda} \quad
		\subjectto \; \boldsymbol{I} - \sum_{k=1}^{K} \boldsymbol{h}_k^H \boldsymbol{h}_k {\lambda}_k \in\mathcal{S}^{N_\text{t}}_{+},
	\label{Eqn:FunctionPsiProof2_Dual2}
\end{equation} 
\noindent where $\rho_k = \phi_k^{-1}(\mu_k)$ and $\boldsymbol{\lambda} \in \mathbb{R}^K$ is a vector of Lagrangian multipliers associated with the constraints in (\ref{Eqn:FunctionPsiProof2_1}).
Since problem (\ref{Eqn:FunctionPsiProof2_1}) is convex and satisfies Slater's conditions, the duality gap is zero and the solutions $\boldsymbol{X}^*$ and $\boldsymbol{\lambda}^*$ satisfy $\psi(\boldsymbol{\mu}) = \Tr{\boldsymbol{X}^*} = \boldsymbol{\rho}^\top \boldsymbol{\lambda}^*$.

Finally, we consider the constraint in (\ref{Eqn:FunctionPsiProof2_Dual2}).
We note that the constraint is satisfied if and only if all the eigenvalues of matrix $\boldsymbol{I} - \boldsymbol{H}^H \boldsymbol{\Lambda} \boldsymbol{H}$, where $\boldsymbol{\Lambda} = \diag\{\boldsymbol{\lambda}\}$, are non-negative, i.e., $\| \boldsymbol{H}^H \boldsymbol{\Lambda} \boldsymbol{H} \|_2 \leq 1$.
Furthermore, $\| \boldsymbol{H}^H \boldsymbol{\Lambda} \boldsymbol{H} \|_2 = \sqrt{ \lambda_{\text{max}} \big( \boldsymbol{\Sigma} \boldsymbol{U}^H \boldsymbol{\Lambda} \boldsymbol{U} \boldsymbol{\Sigma}^2 \boldsymbol{U}^H \boldsymbol{\Lambda} \boldsymbol{U} \boldsymbol{\Sigma} \big) } =$ \linebreak$\sqrt{ \lambda_{\text{max}} \big( \boldsymbol{U} \boldsymbol{\Sigma}^2 \boldsymbol{U}^H \boldsymbol{\Lambda} \boldsymbol{U} \boldsymbol{\Sigma}^2 \boldsymbol{U}^H \boldsymbol{\Lambda} \big) } = \| \boldsymbol{B} \boldsymbol{\Lambda} \boldsymbol{B} \|_2$, where $\boldsymbol{H} = \boldsymbol{U} \boldsymbol{\Sigma} \boldsymbol{V}^H$ is the singular value decomposition of $\boldsymbol{H}$ and $\lambda_{\text{max}}(\boldsymbol{A})$ denotes the largest eigenvalue of $\boldsymbol{A}$.
Hence, the constraint in (\ref{Eqn:FunctionPsiProof2_Dual2}) is equivalent to $\| \boldsymbol{B} \boldsymbol{\Lambda} \boldsymbol{B} \|_2 \leq 1$.
This concludes the proof.

	\section{Proof of Proposition \ref{Theorem:MuProp3}}
	\label{Appendix:Prop3}
	For any given $\bar{\tau}$, $N$, and set of vectors $\boldsymbol{\mu}_n, n\in\{1,2,\cdots, N\}$, optimization problem (\ref{Eqn:OptimalResourceAlloc}) is linear in $\boldsymbol{\tau}$ and involves $N$ variables and $K+1$ constraints.
We note that the solution of a linear optimization problem with $\tilde{N}$ variables and $\tilde{K} \leq \tilde{N}$ constraints is a vertex of the polytope defined by the $\tilde{K}$ constraints, and thus, the corresponding vector determining the solution of the linear problem has at most $\tilde{K}$ non-zero elements \cite{Polik2010}.
Thus, for the optimal solution of (\ref{Eqn:OptimalResourceAlloc}) and, hence, (\ref{Eqn:RefProblem}) and (\ref{Eqn:OriginalProblem}), at most $N^* = K+1$ time slots have non-zero lengths. 
This concludes the proof.

	\section{Proof of Lemma \ref{Theorem:Lemma}}
	\label{Appendix:LemmaProof}
	The proof below follows along the lines of the proof in \cite[Appendix~C]{Shanin2021b}.
Since problem (\ref{Eqn:LemmaProblem}) is convex and Slater's conditions are satisfied, strong duality holds and the gap between problem (\ref{Eqn:LemmaProblem}) and its dual problem is equal to zero \cite{Boyd2004}.
We express the Lagrangian of (\ref{Eqn:LemmaProblem}) as follows:
\vspace*{-7pt}
\begin{equation}
	\vspace*{-10pt}
	\mathcal{L}(\boldsymbol{X}) =  \Tr{\boldsymbol{X}} - \sum_{k=1}^{K} \gamma_k \boldsymbol{h}_k \boldsymbol{X} \boldsymbol{h}_k^H - \boldsymbol{Y} \boldsymbol{X}+ \bar{\gamma},
	\label{Eqn:RankProofLagrangian}
\end{equation}
\noindent where ${\gamma_k}, k \in \{1,2, \cdots, K\}$, are the Lagrangian multipliers associated with the $K$ constraints of problem (\ref{Eqn:LemmaProblem}) and $\bar{\gamma}$ collects all terms that do not depend on $\boldsymbol{X}$.
Here, $\boldsymbol{Y}$ is the Lagrangian multiplier associated with constraint $\boldsymbol{X} \in \mathcal{S}^{N_\text{t}}_{+}$.
We note that the Karush-Kuhn-Tucker (KKT) conditions are satisfied for the optimal solution $\boldsymbol{X}^*$ of (\ref{Eqn:LemmaProblem}) and the solutions ${\gamma}_k^*, k\in\{1,2,\cdots, K\}$, and $\boldsymbol{Y}^*$ of the corresponding dual problem.
The KKT conditions are given by \cite{luo_chang_2009}
\vspace*{-7pt}
\begin{subequations}
	\begin{align}
		&\triangledown \mathcal{L}(\boldsymbol{X}^*) = \boldsymbol{0}_{N_\text{t}\times N_\text{t} }
		\label{Eqn:RankProofKKTGradient}\\
		&\boldsymbol{Y}^* \succeq \boldsymbol{0}_{N_\text{t}\times N_\text{t} }, \gamma_k^* \geq 0, \forall k
		\label{Eqn:RankProofKKTPositive}\\
		& \boldsymbol{Y}^* \boldsymbol{X}^* = \boldsymbol{0}_{N_\text{t}\times N_\text{t} },
		\label{Eqn:RankProofKKTSemidef}
		\vspace*{-10pt}
	\end{align}
	\label{Eqn:RankProofKKT}
\end{subequations}
\noindent\hspace*{-4pt}where $\boldsymbol{0}_{N_\text{t}\times N_\text{t} }$ stands for the all-zero matrix of size $N_\text{t}\times N_\text{t}$ and $\triangledown \mathcal{L}(\boldsymbol{X}^*)$ denotes the gradient of $\mathcal{L}(\boldsymbol{X})$ evaluated at $\boldsymbol{X}^*$.
Next, we express condition (\ref{Eqn:RankProofKKTGradient}) as follows
\vspace*{-7pt}
\begin{equation}
	\vspace*{-10pt}
	\boldsymbol{Y}^* = \boldsymbol{I}_{N_\text{t}} - \boldsymbol{\Delta}, \label{Eqn:RankProofMainEquality} 
\end{equation}
\noindent where $\boldsymbol{\Delta} = \sum_k \gamma_k^* \boldsymbol{h}_k^H \boldsymbol{h}_k$.
Let us now investigate the structure of $\boldsymbol{\Delta}$.
We denote the maximum eigenvalue of $\boldsymbol{\Delta}$ by $\delta^\text{max} \in \mathbb{R}$.
Due to the randomness of the channel, with probability $1$, only one eigenvalue of $\boldsymbol{\Delta}$ has value\footnotemark\hspace*{0pt} $\delta^\text{max}$ \cite{Xu2019}.
Considering (\ref{Eqn:RankProofMainEquality}), we note that if $\delta^\text{max}<1$, then $\boldsymbol{Y}^*$ is a full-rank positive definite matrix.
In this case, (\ref{Eqn:RankProofKKTSemidef}) yields $\boldsymbol{X}^* = \boldsymbol{0}_{N_\text{t}\times N_\text{t} }$, and hence, $\rank \{\boldsymbol{X}^*\}= 0$.
Furthermore, if $\delta^\text{max}>1$, then $\boldsymbol{Y}^*$ is not a positive semidefinite matrix, which contradicts (\ref{Eqn:RankProofKKTPositive}).
Finally, if $\delta^\text{max} = 1$, then $\boldsymbol{Y}^*$ is a positive semidefinite matrix with $\rank \{\boldsymbol{Y}^*\} = N_\text{t} - 1$.
Then, applying Sylvester's rank inequality to (\ref{Eqn:RankProofKKTSemidef}), we have
\vspace*{-7pt}
\begin{equation}
	\vspace*{-7pt}
	0 = \rank \{\boldsymbol{Y}^* \boldsymbol{X}^*\} \geq \rank \{\boldsymbol{Y}^*\} + \rank \{\boldsymbol{X}^*\} - N_\text{t} = \rank \{\boldsymbol{X}^*\} - 1.
\end{equation}
Thus, we have $\rank \{\boldsymbol{X}^*\} \leq 1$. This concludes the proof.
\footnotetext{In our exhaustive simulations, we have always observed that only one eigenvalue of matrix $\boldsymbol{\Delta}$ has value $\delta^\text{max}$.}

	\section{Proof of Proposition \ref{Theorem:MassiveMIMO}}
	\label{Appendix:Prop5}
	Since $\boldsymbol{h}_i \boldsymbol{h}_j^H = 0, \forall i \neq j$, matrix $\boldsymbol{B}$ in Proposition~\ref{Theorem:MuProp2} is diagonal and the solution of problem (\ref{Eqn:FunctionPsiProposition}) is given by $\lambda_k = \frac{1}{\|\boldsymbol{h}_k \|_2^2}, k\in\{1,2,\cdots, K\}$.
Hence, the corresponding instantaneous transmit power can be obtained as $\psi(\boldsymbol{\mu}) = \sum_{k=1}^{K} \phi_k^{-1}(\mu_k) \frac{1}{\| \boldsymbol{h}_k \|^2_2}$. 
Thus, we note that, for the optimal downlink energy signal, the channel $\boldsymbol{H}$ between the BS and the user devices can be equivalently decomposed into $K$ parallel independent subchannels $\boldsymbol{h}_k, k\in\{1,2,\cdots, K\},$ and the power harvested at user device $k$ in time slot $n$ does not depend on the powers harvested at the other devices.
Furthermore, similarly, the optimal average transmit power in the downlink for energy transmission to user $k$ is independent of the powers transmitted to the other users.
Thus, the optimal average transmit power in (\ref{Eqn:ProblemRef1}) can be expressed as $P^*_\text{DL}(\bar{{\tau}}) = \sum_{n=1}^{\bar{N}} \tau^\text{d*}_{n} \psi(\boldsymbol{\mu}^*_n) = \sum_{k=1}^{K} \sum_{n=1}^{\bar{N}} \bar{\tau}^\text{d*}_{n} \phi_k^{-1}(\mu^*_{n,k}) \frac{1}{\| \boldsymbol{h}_k \|^2_2} = \sum_{k=1}^{K} \frac{1}{\| \boldsymbol{h}_k \|^2_2}  P^*_{\text{DL}, k}(\bar{{\tau}})$, where $P^*_{\text{DL}, k}(\bar{{\tau}}) = \sum_{n=1}^{N_k^*} \tilde{\tau}^\text{d*}_{n,k} \phi_k^{-1}(\mu^*_{n,k})$ is the normalized optimal downlink average transmit power for energy transfer to user $k$.
Here, $\tilde{\tau}^\text{d*}_{n,k}$ and $N_k^*$ are the duration of the time slot, where user $k$ harvests power $\mu^*_{n,k}$, $n \in \{1,2,\cdots, N_k^*\}, k \in \{1,2,\cdots, K\}$,  and the optimal number of time slots for user $k$, respectively.

To determine $N_k^*, \tilde{\tau}^\text{d*}_{n,k}$, and $\mu^*_{n,k}, n \in\{1,2,\cdots, N_k^*\}, k\in\{1,2,\cdots, K\}$, we equivalently decompose problem (\ref{Eqn:OptimalResourceAlloc}) into $K$ subproblems, where subproblem $k, k \in \{1,2,\cdots, K\},$ defines the optimal resource allocation for user $k$, and formulate the subproblem for user $k$ as follows:
	\begin{equation}
		\minimize_{ \tilde{\boldsymbol{\tau}}^\text{d}_{k} \succeq \boldsymbol{0}, \tilde{\boldsymbol{\mu}}_{k}, N_k }\, \sum_{n=1}^{N_k} \tilde{\tau}^\text{d}_{n, k} \phi_k^{-1}({\mu}_{n, k}) \quad
		\subjectto \; \sum_{n=1}^{N_k} \tilde{\tau}^\text{d}_{n,k} {\mu}_{n, k} \geq {\xi}^\text{d}_k(\bar{\tau}), \;
		\sum_{n=1}^{N_k} \tilde{\tau}^\text{d}_{n, k} = \bar{\tau},
	\label{Eqn:OptimalResourceAllocSubproblem}
\end{equation}
\noindent\hspace*{0pt}where $\tilde{\boldsymbol{\tau}}^\text{d}_{k} = [\tilde{\tau}^\text{d}_{1, k}, \tilde{\tau}^\text{d}_{2, k}, \cdots, \tilde{\tau}^\text{d}_{N_k, k}]$ and $\tilde{\boldsymbol{\mu}}_{k} = [{\mu}_{1, k}, {\mu}_{2, k}, \cdots, {\mu}_{N_k, k}]$.
Similar to the single-user WPCN in Proposition~\ref{Theorem:SingleUser}, at most $\bar{N}_k = 2$ time slots are needed for problem (\ref{Eqn:OptimalResourceAllocSubproblem}).
The optimal harvested powers and the durations of the corresponding time slots are given by ${\mu}^*_{1,k} = \phi(A_k^2)$, ${\mu}^*_{2,k} = 0$ and $\tilde{\tau}^\text{d}_{1,k} = \frac{\xi^\text{d}_k(\bar{\tau})}{A_k^2}$, $\tilde{\tau}^\text{d}_{2,k} = \bar{\tau} - \tilde{\tau}^\text{d}_{1,k}$, respectively.
Hence, the optimal average transmit power in (\ref{Eqn:ProblemRef1}) is given by $P^*_\text{DL}(\bar{{\tau}}) = \sum_{k=1}^{K} \frac{\xi^\text{d}_k(\bar{\tau})}{\| \boldsymbol{h}_k \|^2_2}$.
We note that the downlink transmit policy in Proposition~\ref{Theorem:MassiveMIMO} yields the optimal transmit power $P^*_\text{DL}(\bar{{\tau}})$ and satisfies the constraints in (\ref{Eqn:OptimalResourceAllocSubproblem}).
Hence, for a given $\bar{{\tau}}$, the energy beamforming vectors and time slot durations in (\ref{Eqn:SolutionMassiveMimo}) are the solution of (\ref{Eqn:RefProblem}).
This concludes the proof.

\bibliographystyle{IEEEtran}
\bibliography{WPCN_main.bib}

\begin{thebibliography}{10}
\providecommand{\url}[1]{#1}
\csname url@samestyle\endcsname
\providecommand{\newblock}{\relax}
\providecommand{\bibinfo}[2]{#2}
\providecommand{\BIBentrySTDinterwordspacing}{\spaceskip=0pt\relax}
\providecommand{\BIBentryALTinterwordstretchfactor}{4}
\providecommand{\BIBentryALTinterwordspacing}{\spaceskip=\fontdimen2\font plus
\BIBentryALTinterwordstretchfactor\fontdimen3\font minus
  \fontdimen4\font\relax}
\providecommand{\BIBforeignlanguage}[2]{{%
\expandafter\ifx\csname l@#1\endcsname\relax
\typeout{** WARNING: IEEEtran.bst: No hyphenation pattern has been}%
\typeout{** loaded for the language `#1'. Using the pattern for}%
\typeout{** the default language instead.}%
\else
\language=\csname l@#1\endcsname
\fi
#2}}
\providecommand{\BIBdecl}{\relax}
\BIBdecl

\bibitem{Shanin2021d}
N.~Shanin, M.~Garkisch, A.~Hagelauer, R.~Schober, and L.~Cottatellucci,
  ``Optimal resource allocation and beamforming for two-user {MISO WPCNs} for a
  non-linear circuit-based {EH} model,'' in \emph{Proc. {IEEE} Int. Conf.
  Acoust., Speech and Signal Process. ({ICASSP})}, 2022.

\bibitem{Bi2015}
S.~Bi, C.~K. Ho, and R.~Zhang, ``Wireless powered communication: opportunities
  and challenges,'' \emph{{IEEE} Commun. Mag.}, vol.~53, no.~4, pp. 117--125,
  Apr. 2015.

\bibitem{Clerckx2019}
B.~{Clerckx}, R.~{Zhang}, R.~{Schober}, D.~W.~K. {Ng}, D.~I. {Kim}, and H.~V.
  {Poor}, ``Fundamentals of wireless information and power transfer: From {RF}
  energy harvester models to signal and system designs,'' \emph{IEEE J. Sel.
  Areas Commun.}, vol.~37, no.~1, pp. 4--33, Jan. 2019.

\bibitem{Ju2014}
H.~Ju and R.~Zhang, ``Throughput maximization in wireless powered communication
  networks,'' \emph{{IEEE} Trans. Wirel. Commun.}, vol.~13, no.~1, pp.
  418--428, Jan. 2014.

\bibitem{Liu2014}
L.~Liu, R.~Zhang, and K.-C. Chua, ``Multi-antenna wireless powered
  communication with energy beamforming,'' \emph{{IEEE} Trans. Commun.},
  vol.~62, no.~12, pp. 4349--4361, Dec. 2014.

\bibitem{Yang2015}
G.~Yang, C.~K. Ho, R.~Zhang, and Y.~L. Guan, ``Throughput optimization for
  massive {MIMO} systems powered by wireless energy transfer,'' \emph{{IEEE} J.
  Sel. Areas Commun.}, vol.~33, pp. 1640--1650, Jan. 2015.

\bibitem{Lee_2016}
H.~Lee, K.-J. Lee, H.~Kim, B.~Clerckx, and I.~Lee, ``Resource allocation
  techniques for wireless powered communication networks with energy storage
  constraint,'' \emph{{IEEE} Trans. Wirel. Commun.}, vol.~15, no.~4, pp.
  2619--2628, Apr. 2016.

\bibitem{Morsi2018a}
R.~Morsi, D.~S. Michalopoulos, and R.~Schober, ``Performance analysis of
  near-optimal energy buffer aided wireless powered communication,''
  \emph{{IEEE} Trans. Wirel. Commun.}, vol.~17, no.~2, pp. 863--881, Nov. 2017.

\bibitem{Gu2022}
X.~Gu, S.~Hemour, and K.~Wu, ``Far-field wireless power harvesting: Nonlinear
  modeling, rectenna design, and emerging applications,'' \emph{Proc. IEEE},
  vol. 110, no.~1, pp. 56--73, Jan. 2022.

\bibitem{Wei2022}
Z.~Wei, X.~Yu, D.~W.~K. Ng, and R.~Schober, ``Resource allocation for
  simultaneous wireless information and power transfer systems: A tutorial
  overview,'' \emph{Proc. IEEE}, vol. 110, no.~1, pp. 127--149, Jan. 2022.

\bibitem{Tietze2012}
U.~Tietze and C.~Schenk, \emph{{Advanced Electronic Circuits}}.\hskip 1em plus
  0.5em minus 0.4em\relax Springer Science \& Business Media, 2012.

\bibitem{Shanin2020}
N.~Shanin, L.~Cottatellucci, and R.~Schober, ``Markov decision process based
  design of {SWIPT} systems: Non-linear {EH} circuits, memory, and impedance
  mismatch,'' \emph{{IEEE} Trans. Commun.}, vol.~69, no.~2, pp. 1259 -- 1274,
  Feb. 2021.

\bibitem{Boshkovska2015}
E.~{Boshkovska}, D.~W.~K. {Ng}, N.~{Zlatanov}, and R.~{Schober}, ``Practical
  non-linear energy harvesting model and resource allocation for {SWIPT}
  systems,'' \emph{IEEE Commun. Lett.}, vol.~19, no.~12, pp. 2082--2085, Dec.
  2015.

\bibitem{Boshkovska2017a}
E.~Boshkovska, D.~W.~K. Ng, N.~Zlatanov, A.~Koelpin, and R.~Schober, ``Robust
  resource allocation for {MIMO} wireless powered communication networks based
  on a non-linear {EH} model,'' \emph{{IEEE} Trans. Commun.}, vol.~65, no.~5,
  pp. 1984--1999, May 2017.

\bibitem{Boshkovska2018}
E.~Boshkovska, D.~W.~K. Ng, L.~Dai, and R.~Schober, ``Power-efficient and
  secure {WPCNs} with hardware impairments and non-linear {EH} circuit,''
  \emph{{IEEE} Trans. Commun.}, vol.~66, no.~6, pp. 2642--2657, Jun. 2018.

\bibitem{Hua2022}
M.~Hua, Q.~Wu, and H.~V. Poor, ``Power-efficient passive beamforming and
  resource allocation for {IRS}-aided {WPCNs},'' \emph{{IEEE} Trans. Commun.},
  vol.~70, no.~5, pp. 3250--3265, Mar. 2022.

\bibitem{Zeng2021}
P.~Zeng, Q.~Wu, and D.~Qiao, ``Energy minimization for {IRS}-aided {WPCNs} with
  non-linear energy harvesting model,'' \emph{{IEEE} Wirel. Commun. Lett.},
  vol.~10, no.~11, pp. 2592--2596, Nov. 2021.

\bibitem{Li2022}
Z.~Li, W.~Chen, Q.~Wu, H.~Cao, K.~Wang, and J.~Li, ``Robust beamforming design
  and time allocation for {IRS}-assisted wireless powered communication
  networks,'' \emph{{IEEE} Trans. Commun.}, vol.~70, no.~4, pp. 2838--2852,
  Feb. 2022.

\bibitem{Clerckx2015}
B.~Clerckx, E.~Bayguzina, D.~Yates, and P.~D. Mitcheson, ``Waveform
  optimization for wireless power transfer with nonlinear energy harvester
  modeling,'' in \emph{Proc. Int. Symp. Wirel. Commun. Syst. ({ISWCS})}, Aug.
  2015.

\bibitem{Clerckx2018}
B.~{Clerckx}, ``Wireless information and power transfer: Nonlinearity, waveform
  design, and rate-energy tradeoff,'' \emph{IEEE Trans. Signal Process.},
  vol.~66, no.~4, pp. 847--862, Feb. 2018.

\bibitem{Clerckx2017}
B.~Clerckx, Z.~B. Zawawi, and K.~Huang, ``Wirelessly powered backscatter
  communications: Waveform design and {SNR}-energy tradeoff,'' \emph{{IEEE}
  Commun. Lett.}, vol.~21, no.~10, pp. 2234--2237, Oct. 2017.

\bibitem{Zawawi2019}
Z.~B. Zawawi, Y.~Huang, and B.~Clerckx, ``Multiuser wirelessly powered
  backscatter communications: Nonlinearity, waveform design, and {SINR}-energy
  tradeoff,'' \emph{{IEEE} Trans. Wirel. Commun.}, vol.~18, no.~1, pp.
  241--253, Jan. 2019.

\bibitem{Shanin2021a}
N.~Shanin, L.~Cottatellucci, and R.~Schober, ``Optimal transmit strategy for
  multi-user {MIMO} {WPT} systems with non-linear energy harvesters,''
  \emph{{IEEE} Trans. Commun.}, vol.~70, no.~3, pp. 1726--1741, Mar. 2022.

\bibitem{Morsi2019}
R.~{Morsi}, V.~{Jamali}, A.~{Hagelauer}, D.~W.~K. {Ng}, and R.~{Schober},
  ``Conditional capacity and transmit signal design for {SWIPT} systems with
  multiple nonlinear energy harvesting receivers,'' \emph{IEEE Trans. Commun.},
  vol.~68, no.~1, pp. 582--601, Jan. 2020.

\bibitem{Interdonato2020}
G.~Interdonato, M.~Karlsson, E.~Bjornson, and E.~G. Larsson, ``Local partial
  zero-forcing precoding for cell-free massive {MIMO},'' \emph{{IEEE} Trans.
  Wirel. Commun.}, vol.~19, no.~7, pp. 4758--4774, Jul. 2020.

\bibitem{Grant2015}
M.~Grant and S.~Boyd, ``{CVX}: Matlab software for disciplined convex
  programming, version 2.0 beta (2013),'' \emph{URL: http://cvxr. com/cvx},
  2015.

\bibitem{Polik2010}
I.~P{\'o}lik and T.~Terlaky, \emph{Interior Point Methods for Nonlinear
  Optimization}.\hskip 1em plus 0.5em minus 0.4em\relax Springer, 2010.

\bibitem{Shanin2021b}
N.~Shanin, L.~Cottatellucci, and R.~Schober, ``Optimal transmit strategy for
  {MIMO} {WPT} systems with non-linear energy harvesting,'' in \emph{Proc. IEEE
  Int. Conf. Distrib. Comput. Sens. Syst. ({DCOSS})}.\hskip 1em plus 0.5em
  minus 0.4em\relax {IEEE}, 2021.

\bibitem{Boyd2004}
S.~Boyd, S.~P. Boyd, and L.~Vandenberghe, \emph{Convex Optimization}.\hskip 1em
  plus 0.5em minus 0.4em\relax Cambridge University Press, 2004.

\bibitem{luo_chang_2009}
Z.-Q. Luo and T.-H. Chang, \emph{SDP Relaxation of Homogeneous Quadratic
  Optimization: Approximation Bounds and Applications}.\hskip 1em plus 0.5em
  minus 0.4em\relax Cambridge University Press, p. 117–165.

\bibitem{Xu2019}
D.~Xu, X.~Yu, Y.~Sun, D.~W.~K. Ng, and R.~Schober, ``Resource allocation for
  secure {IRS}-assisted multiuser {MISO} systems,'' in \emph{Proc. {IEEE}
  Globecom Workshop}, Dec. 2019.

\end{thebibliography}

\end{document}